\documentclass[journal,twocolumn]{IEEEtran}
\usepackage[utf8]{inputenc} 
\usepackage[T1]{fontenc}
\usepackage{url}
\usepackage{ifthen}
\usepackage{derivative}
\usepackage{cite}
\usepackage[cmex10]{amsmath}
\usepackage{amsmath,amssymb,amsfonts,amsthm}
\usepackage{graphicx}
\usepackage[short]{optidef}
\newtheorem{theorem}{Theorem}
\newtheorem{lemma}{Lemma}
\usepackage{algorithm,algorithmic}
\usepackage{xcolor}
\usepackage{physics}
\newtheorem{corollary}{Corollary}
\newtheorem*{remark}{Remark}
\usepackage{multicol}
\usepackage{caption} 
\usepackage{dsfont}
\usepackage{subcaption} 

\renewcommand{\d}{\mathrm{d}}
\newcommand{\defn}{\stackrel{\triangle}{=} }
\usepackage[hidelinks,citecolor={blue}]{hyperref}

\newcommand{\bs}{\boldsymbol}

\renewcommand{\d}{\mathrm{d}}

\begin{document}
\title{Information-Energy Capacity Region for SLIPT Systems over Lognormal Fading Channels: A Theoretical and Learning-Based Analysis} 
\author{Nizar~Khalfet,~\IEEEmembership{Member,~IEEE,} Kapila~W.~S.~Palitharathna,~\IEEEmembership{Member,~IEEE,}
Symeon~Chatzinotas,~\IEEEmembership{Fellow,~IEEE,} Ioannis~Krikidis,~\IEEEmembership{Fellow,~IEEE}\vspace{-11mm}
\thanks{N. Khalfet, K. W. S. Palitharathna, and I. Krikidis are with the IRIDA Research Centre for Communication Technologies, Department of Electrical and Computer Engineering, University of Cyprus, 1678 Nicosia, Cyprus (e-mail: \{khalfet.nizar, palitharathna.kapila, krikidis\}@ucy.ac.cy).} 
\thanks{S. Chatzinotas is with the Interdisciplinary Centre for Security, Reliability and Trust (SnT), University of Luxembourg, L-1855 Luxembourg City, Luxembourg (email: symeon.chatzinotas@uni.lu)}\thanks{This work received funding from the European Union’s Horizon-JU-SNS research and innovation programme under the projects iSEE-6G (Grant agreement No. 101139291) and from the European Union’s Erasmus+ Programme under the CoRe 5G project (Grant Agreement No. 101260846).}}
\maketitle
\begin{abstract}
This paper presents a comprehensive analysis of the information-energy capacity region for simultaneous lightwave information and power transfer (SLIPT) systems over lognormal fading channels. Unlike conventional studies that primarily focus on additive white Gaussian noise channels, we study the complex impact of lognormal fading, which is prevalent in optical wireless communication systems such as underwater and atmospheric channels. By applying the Smith's framework to these channels, we demonstrate that the optimal input distribution is discrete, characterized by a finite number of mass points. We further investigate the properties of these mass points, especially at the transition points, to reveal critical insights into the rate-power trade-off inherent in SLIPT systems. Additionally, we introduce a novel cooperative information-energy capacity learning framework, leveraging generative adversarial networks, to effectively estimate and optimize the information-energy capacity region under practical constraints. Numerical results validate our theoretical findings, illustrating the significant influence of channel fading on system performance. The insights and methodologies presented in this work provide a solid foundation for the design and optimization of future SLIPT systems operating in challenging environments.
\end{abstract}


\vspace{-7mm}
\section{Introduction}
Optical wireless communication (OWC) has been recognized as a promising technology to achieve high-speed, low-latency, and highly secure communication. By exploiting the optical spectrum, OWC overcomes the spectrum limitations of radio frequency technology and is therefore suitable for a wide range of applications in future wireless systems, including 6G networks. OWC encompasses several technologies depending on the application environment, such as free-space optical communication, visible light communication, light-fidelity, underwater optical wireless communication, and non-terrestrial satellite communication~\cite{Chowdhury_2018}. In these systems, light-emitting diodes (LEDs) or laser diodes (LDs) are used as transmitters, while photodiodes (PDs) or photovoltaic (PV) cells serve as receivers. Moreover, existing lighting infrastructures can be easily adapted to simultaneously support illumination, information transfer, and lightwave power transfer~\cite{Chowdhury_2018}.

Recently, simultaneous lightwave information and power transfer (SLIPT) has emerged as a new communication paradigm that enables the concurrent transmission of information and energy using optical signals~\cite{Ding_2018, Vasilis_2023}. This paradigm is particularly relevant for energy-constrained devices and sensors, especially in scenarios where conventional wireless power transfer techniques are inefficient or infeasible, such as underwater, inter-satellite, vehicular, and indoor systems~\cite{Vasilis_2023}. Prior works have investigated various aspects of SLIPT systems, including energy maximization under quality-of-service constraints~\cite{Ding_2018}, rate and power optimization~\cite{Zhang_2019}, receiver architectures and policies~\cite{Geofeng_2019}, and rate-power trade-offs~\cite{Sahand_2021}. In addition, SLIPT has been applied to several practical scenarios, such as multi-cell beamforming systems~\cite{Abdelhady_2020}, resonant beam energy transfer~\cite{Liu_2021}, and underwater communication systems~\cite{Sait_2019}.

Despite these advances, most existing studies assume simplified channel models, typically additive white Gaussian noise (AWGN) channels with deterministic path loss. Consequently, the impact of channel fading on the information-energy trade-off, as well as on the structure of the optimal input distribution, remains largely unexplored. In practical OWC systems, fading effects are non-negligible and are commonly modeled using distributions such as lognormal, gamma-gamma, generalized gamma, and exponential gamma~\cite{Jamali_2018}. Among these, the lognormal fading model is particularly suitable for weak turbulence conditions in both underwater and atmospheric environments~\cite{Jamali_2018, Yang_2016}. Therefore, investigating SLIPT systems under lognormal fading is essential for realistic system design.

From an information-theoretic perspective, closed-form expressions for the exact capacity of OWC channels are generally unavailable, and existing works instead rely on upper and lower bounds~\cite{Chaaban_2017, Lapidoth_2009}. Furthermore, the information-energy capacity region of SLIPT systems has received limited attention, with only partial characterizations available in the literature~\cite{Nikita_2023}. In particular, the joint impact of lognormal channel fading, practical power constraints, and energy harvesting requirements on the optimal input distribution of SLIPT systems remains an important open problem.

OWC transmitters are subject to peak-power (PP) and average-power (AP) constraints due to device linearity limits and safety requirements~\cite{Chaaban_2017, Lapidoth_2009}. In SLIPT systems, an additional energy harvesting (EH) constraint must be satisfied to ensure sufficient power delivery at the receiver. While most existing works adopt linear EH models for analytical tractability~\cite{Zhang_2019}, practical EH circuits exhibit nonlinear behavior~\cite{Nikita_2023, Sait_2019}. This nonlinearity introduces additional complexity in characterizing the information-energy trade-off and significantly affects the structure of the capacity-achieving input distribution. To the best of the authors’ knowledge, a rigorous information-theoretic analysis of SLIPT systems that jointly incorporate PP, AP, and nonlinear EH constraints over lognormal fading channels remains lacking. Although~\cite{Chan_2005} establishes results for conditionally Gaussian channels with AP constraints, it is not directly applicable to our setting due to fundamental differences in the channel model and constraints. In particular, our work considers a lognormal fading channel and incorporates additional nonlinear EH constraints, leading to a more complex capacity characterization. Similarly, while~\cite{Duman_2018} studies discreteness of capacity-achieving distributions under symmetric amplitude constraints and bounded fading, our system employs intensity modulation with non-negative inputs and lognormal fading with unbounded support, resulting in different structural properties.

On the other hand, recent advances in data-driven methods have demonstrated the potential of deep learning (DL) in physical layer communications~\cite{Data}, with successful applications to signal detection, decoding, and channel estimation. However, their use in fundamental information-theoretic problems, such as channel capacity characterization, remains relatively underexplored~\cite{DataEstim}. The cooperative channel capacity learning (CORTICAL) framework introduced in~\cite{Letizia_2023} represents a notable step in this direction, leveraging a generative adversarial network (GAN)-based architecture to learn capacity-achieving input distributions. This framework employs a GAN-inspired cooperative approach where two neural networks operate in tandem: a generator that learns the capacity-achieving input distribution and a discriminator that distinguishes paired from unpaired channel samples. However, despite progress in neural network-based mutual information estimation, it remains unclear whether such approaches provide new insights into capacity regions for complex systems, such as SLIPT systems under lognormal fading with nonlinear EH constraints.

In this work, we consider a point-to-point SLIPT system over a lognormal fading channel under a non-coherent setting, where neither the transmitter nor the receiver has access to channel state information (CSI)~\cite{Faycal_2001}. The `no-CSI' assumption establishes a worst-case performance bound and provides a fundamental limit for systems operating without channel adaptation~\cite{Faycal_2001}. It is particularly relevant for practical low-complexity SLIPT receivers, such as those in Internet of Things and sensor applications, where the energy and computational overhead of channel estimation is prohibitive. These features distinguish our work from prior studies on SWIPT systems under Rayleigh fading with receiver-side CSI~\cite{Khalfet_2023}. In this framework, the system performance is characterized through the ergodic capacity of the averaged channel, where the transition probability is given by $p(y|x)=\mathbb{E}_{h}[p(y|x,h)]$. This non-coherent formulation differs fundamentally from CSI-based approaches~\cite{Duman_2018}, where mutual information is conditioned on the instantaneous channel realization.

This paper studies the information-energy capacity region of SLIPT systems where a single transmitter communicates with an information receiver while simultaneously delivering energy to an EH receiver over lognormal fading channels under practical system constraints. The system is subject to PP, AP, and nonlinear EH constraints, modeled through a function of the form $\mathbb{E}[bX \ln(1+CX)]$, which captures the characteristics of photovoltaic EH. Our objective is not to alter hardware characteristics such as LED linearity or PV cell efficiency, but to optimize the input distribution within these inherent constraints to characterize the achievable information-energy trade-off. Results demonstrate that adapting the input distribution to channel fading and EH nonlinearity significantly enhances system performance, highlighting the practicality and effectiveness of the proposed framework.

An explicit expression for the transition probability distribution of the lognormal fading channel is derived. By leveraging Smith's framework~\cite{Smith_1971} and employing a Hermite polynomial basis, we establish that the optimal input distribution is discrete with a finite number of mass points~\cite{Kapila_2024}. We further characterize how this distribution evolves across regimes, including high-SNR conditions and stringent PP constraints, and propose an achievable input distribution for the high-SNR regime. In particular, we identify critical transitions under SLIPT constraints, where binary inputs cease to be optimal, and mass points shift toward higher amplitudes due to nonlinear EH requirements. By incorporating a nonlinear EH constraint, we show how the feasible input space becomes non-convex, revealing a fundamental trade-off between information transfer and EH, and providing new insights into the structure of capacity-achieving distributions beyond conventional lognormal and Gaussian channel models.

In addition, we propose a Cooperative Information-Energy Capacity Learning (CIECL) framework based on GANs to optimize SLIPT systems under practical constraints. Unlike traditional methods such as the Blahut-Arimoto (BA) algorithm~\cite{blahut}, which rely on discretization and struggle with nonlinear, non-convex EH constraints, the proposed approach learns discrete, capacity-achieving input distributions directly in a flexible and data-driven manner. The generator network dynamically adapts to system constraints, including PP, AP, and nonlinear EH, without requiring predefined parametric forms or grid-based approximations and can achieve higher mutual information than BA algorithm under identical conditions. Compared to conventional approaches, our GAN-based framework (i) avoids the need for discretization, (ii) generalizes across different constraints and channel conditions, and (iii) efficiently captures the interplay between information transfer and EH constraints. Unlike the conventional GAN-based capacity learning framework in~\cite{Letizia_2023}, which focuses on generic point-to-point communication systems, our approach is uniquely designed for SLIPT systems, by introducing modifications to the value function to explicitly enforce PP, AP, and EH constraints, and extending the learning model to handle dual-receiver scenarios. We further modify the activation function in the generator’s output layer, applying a modified version of rectified linear unit (ReLU) that enforces the non-negativity and PP limits of the transmitted signal. Our results demonstrate that the proposed approach effectively captures the structure of optimal input distributions and provides improved performance compared to traditional methods.

\begin{table*}[!t]
\caption{Summary of notation.}
\centering
\begin{tabular}{|p{1.5cm}|p{5.5cm}||p{1.5cm}|p{5.5cm}|  }
\hline
\textbf{Notations} &\textbf{Description}&\textbf{Notations} &\textbf{Description}\\
 \hline
   $N$   &  Number of mass points & $X$ & Random variable of the input signal\\
    \hline
   $t$  & Time index  & $X^{\star}$  & Optimal input random variable \\
    \hline
       $A$  & PP constraint  & $H(X)$  & Entropy of a random variable $X$ \\
    \hline
  $\varepsilon$   &  AP constraint & $F^{(n)}$&   Sequence of the input distribution \\
     \hline
     $E_{\text{th}}$&  Energy required at the EH& $F(\cdot)$ &Probability distribution function of $X$\\
 \hline
 $h_{1}$ & Channel gain for the information link & $I(\cdot)$&   Mutual information as a function of the distribution $F$\\
 \hline
 $h_{1,l}$&  Path loss for the information link & $F^{\star}$&  Optimal input probability distribution\\
 \hline
  $h_{t}$ & Lognormal fading & $G$ & Generator\\
 \hline
$h_{2}$   &Channel gain for the EH link & $D$ & Discriminator \\
 \hline
 $F^{(n)}$ &  Sequence of a distribution $F$  &  $i(x;F)$& Mutual information density of $F$ evaluated at $x$\\
 \hline
 $F_N^{\star}$& Optimal  distribution with $N$ mass points  & $\mathcal{N}(0,\hspace{-0.5ex} \sigma_1^2)$  & Real Gaussian random variable with zero mean and $\sigma_1^2$ variance\\
 \hline
  $\mathrm{Supp}(F)$ & Support of a distribution $F$& $E_0$&  Point of increase of  $F^{\star}$ \\
 \hline
\end{tabular}
\vspace{-5mm}
\end{table*}
The key technical contributions of this paper are summarized as follows.
\begin{itemize}
    \item We characterize the information-energy capacity region of SLIPT systems over lognormal fading channels under PP, AP, and nonlinear EH constraints.
    
    \item We derive the transition probability distribution of the lognormal fading channel and establish, using Smith's framework~\cite{Smith_1971}, that the optimal input distribution is discrete with a finite number of mass points.
    
    \item We analyze the structural properties of the optimal input distribution across different operating regimes, including high-SNR conditions and stringent PP constraints, revealing transitions in the distribution behavior.
    
    \item We propose an achievable input distribution for the high-SNR regime and provide insights into the trade-off between information transmission and energy harvesting.
    
    \item We develop a data-driven CIECL framework based on GANs, which efficiently learns capacity-achieving input distributions under complex system constraints without requiring discretization.
\end{itemize}

\emph{Notation:} In this paper, sets are denoted with uppercase calligraphic letters. Random variables are denoted by uppercase letters, e.g., $X$. The realization and the set of the events from which the random variable $X$ takes values are  denoted by $x$ and $\mathcal{X}$, respectively. The operator $\mathbb{E}$ denotes the expectation with respect to the distribution of a random variable ${X}$. The notation $F$ denotes the probability distribution function of a random variable $X$, and $F^{\star}$ represents the optimal input distribution. The notation $\mathrm{Supp}(F)$ is the support of a distribution $F$, i.e, $\mathrm{Supp}(F)=\{x \in \mathcal{X} | \int_{x \in \mathcal{O}} \d F(x) >0 \hspace{1ex} \text{for every open neighborhood} \hspace{1ex} \mathcal{O} \hspace{1ex} \text{of} \hspace{1ex} x.\}$. $\log_2$ denotes the logarithm in the base $2$, while $\ln$ denotes the logarithm in the base $e$. $||\cdot||_1$ denotes the $L^1$ norm. The probability space is denoted by \( (\mathcal{X}, \mathcal{F}, P) \), where \( \mathcal{X} \) represents the sample space, \( \mathcal{F} \) is the sigma-algebra of measurable subsets of \( \mathcal{X} \), and \( P \) is the probability measure.
     For a random variable \( X \) with a probability measure \( P_X \), the expectation is defined as
    \begin{equation}
        \mathbb{E}[X] = \int_{\mathcal{X}} x dP_X(x).
    \end{equation}
    This integral notation signifies integration with respect to the probability measure \( P_X \).
     The mutual information between input \( X \) and output \( Y \) is given by
    \begin{equation}
        I(X; Y) = \int_{\mathcal{X} \times \mathcal{Y}} p(x, y) \log_2 \frac{p(x, y)}{p_X(x) p_Y(y)} dx dy.
    \end{equation}
    This is equivalent to 
    \begin{equation}
        I(X; Y) = \int_{\mathcal{X} \times \mathcal{Y}} \log_2 \frac{dP_{X, Y}}{d(P_X \times P_Y)} dP_{X,Y},
    \end{equation}
    where \( P_{X,Y} \) denotes the joint probability measure and \( P_X \times P_Y \) represents the product measure.
     The optimal input distribution \( F^* \) is a probability measure on the input space \( \mathcal{X} \), and expectations with respect to this measure are written as
    \begin{equation}
        \mathbb{E}_{F^*}[g(X)] = \int_{\mathcal{X}} g(x) dF^*(x).
    \end{equation}
Table I summarizes the key notation of the paper. The remainder of this paper is structured as follows. In Section \ref{System_model}, we present the system model, channel model, and the problem formulation. The discreteness of the optimal input distribution is proved in Section \ref{discreteness}. In Section \ref{sec:opt_input}, an achievable input distribution for high SNR regime is derived. The properties of the mass points and the analysis of the transition points where the binary input distribution is no longer optimal are presented in Section \ref{sec:mass_points}. The CIECL framework is presented in Section \ref{sec:CIECL}. Finally, numerical results are presented in Section \ref{Results} and Section \ref{conclusion} concludes the paper.

\vspace{-3mm}
\section{System Model}\label{System_model}
\begin{figure}[!t]
    \centering
    \includegraphics[width=0.99\columnwidth]{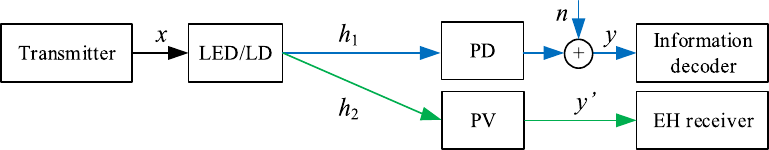}
    \caption{A SLIPT system over a lognormal-fading channel with an LED/LD transmitter, a PD-based information decoder, and a PV cell EH receiver.}
    \label{fig:f1}
    \vspace{-5mm}
\end{figure}

We consider a SLIPT system over a lognormal channel, wherein the transmitter is comprised of a single, narrow-beam LED/LD. A PD receiver is used for information reception and a PV cell receiver is used for lightwave EH, as detailed in previous works~\cite{Vasilis_2023, Geofeng_2019, Ding_2018}. The transmitter uses intensity modulation (IM), and hence, the instantaneous output optical power is proportional to the driving current signal. At the receiver, direct detection (DD) is used, and hence, a current is generated through the PD/PV cell that is proportional to the incident optical power. The generated current signals from the PD and the PV cell are separately sent for information decoding and EH. The harvested energy from the PV cell receiver can be used to charge a battery while the received signal from the PD is first sent through a trans-impedance amplifier to convert it to a voltage signal and then goes through the decoding stage~\cite{Zhang_2019}. A system diagram of this model is depicted in Fig. \ref{fig:f1}.

\vspace{-4mm}
\subsection{Channel Model}
The channel gain from the LED/LD to the PD can be modeled using $h_1 = h_{1,l}h_t$, where $h_{1,l}$ is the path loss, and $h_t$ is the lognormal distributed turbulence-induced fading. The path loss $h_{1,l}$ depends on the transmitter-receiver geometry, and can be expressed as~\cite{Anous_2018}
\vspace{-1mm}
    \begin{align}\label{equ:e5}
        \displaystyle h_{1,l} = \eta_t\eta_{r} e^{\frac{-c(\lambda)l}{\cos{\theta_{PD}}}}\frac{A_{PD}\cos{\theta_{PD}}}{2\pi l^2(1-\cos{\theta_0})},
    \end{align}
where $\eta_t$ and $\eta_r$ are the optical efficiencies of the transmitter and receiver, respectively, $c(\lambda)$ is the extinction coefficient of the channel, $\lambda$ is the optical wavelength, $l$ is the perpendicular distance between the transmitter plane and the receiver plane, $\theta_0$ is the transmitter beam divergence angle, $\theta_{PD}$ is the angle between the perpendicular axis to the receiver plane and the transmitter-PD trajectory, and $A_{PD}$ is the PD aperture area.

The lognormal distribution is used to model the fading coefficient $h_t$ in weak turbulence conditions in OWC (\textit{i.e.}, weak oceanic turbulence and weak atmospheric turbulence). Its probability density function (pdf) is given by~\cite{Jamali_2018, Jamali_2017}
\vspace{-2mm}
    \begin{align}\label{equ:logn}
        p_{h_t}(h_t) = \frac{1}{2h_t\sqrt{2\pi\sigma_{X_l}^2}}\exp\left(-\frac{(\ln {h_t}-2\mu_{X_l})^2}{8\sigma_{X_l}^2}\right).
    \end{align}
The fading coefficient in \eqref{equ:logn} takes the form $h_t=e^{2X_l}$, where $X_l$ is the fading log-amplitude, which is Gaussian distributed with mean $\mu_{X_l}$ and variance $\sigma_{X_l}^2$. The fading coefficient is normalized to ensure that the fading does not affect the average power \textit{i.e.}, $\mu_{X_l} = -\sigma^2_{X_l}$. 
We assume that the channel gain between the LED/LD and the PV is static, and hence, it can be modeled as $h_2 = h_{2,l}$, where $h_{2,l}$ is the geometric path loss from the LED/PD transmitter to the PV cell receiver. It can be obtained by replacing $A_{PD}$ and $\theta_{PV}$ in \eqref{equ:e5} with $A_{PV}$ and $\theta_{PV}$, respectively, where $A_{PV}$ is the effective area of the PV cell and $\theta_{PV}$ is the angle between the perpendicular axis to the receiver plane and the transmitter-PV cell trajectory. In this work, we adopt the no-CSI at either side framework introduced by Abou-Faycal et al.~\cite{Faycal_2001}. Specifically, neither the transmitter nor the receiver has access to instantaneous realizations of the lognormal fading channel. From a practical perspective, the no-CSI assumption captures both theoretical and operational considerations. Theoretically, it provides a fundamental performance bound that does not benefit from channel adaptation, and is widely recognized as a benchmark in noncoherent information theory~\cite{Faycal_2001}. Practically, it reflects the regime of ultra-low-power SLIPT receivers, where the energy and computational overhead required for accurate channel estimation, synchronization, and feedback is too costly. This work considers a point-to-point SLIPT system over a lognormal fading
channel, where neither the transmitter nor the receiver has access to instantaneous CSI [20]. Consequently, the transmitted signal design cannot adapt to specific channel realizations, and the receiver performs non-coherent detection. The channel is therefore modeled as a compound channel, with transition probability obtained by averaging over the fading distribution, i.e.,
    \begin{equation}
        p_{Y|X}(y|x) = \int p_{Y|X,H}(y|x,h)\, f_H(h)\, \mathrm{d}h,
    \end{equation}
where $f_H(h)$ denotes the lognormal fading distribution. The ergodic non-coherent capacity is then defined as
    \begin{equation}
        C = \sup_{F \in \Omega} I(X;Y),
    \end{equation}
where the supremum is taken over all input distributions $F$ that satisfy the PP, 
AP, and nonlinear EH constraints. This formulation is fundamentally different from the receiver-CSI framework adopted in~\cite{Duman_2018}, where the capacity is expressed as $\mathbb{E}_H[I(X;Y|H)]$ 
by conditioning on each channel realization. In contrast, our model strictly follows the no-CSI ergodic capacity framework, where neither the transmitter nor the receiver has access to the fading state \cite{Faycal_2001}, and capacity is defined through the averaged channel law. 

The following lemma establishes the relationship between the ergodic no-CSI capacity
considered in this work and the classical ergodic capacity with receiver CSI.

\noindent\begin{lemma}
Consider a fading channel with channel state $H$ independent of the channel input $X$. Let the averaged channel transition probability be defined as
        \begin{equation}
            p_{Y|X}(y|x)=\int p_{Y|X,H}(y|x,h)f_H(h)dh .
        \end{equation}
Let the ergodic no-CSI capacity be defined as
        \begin{equation}
            C = \sup_{F \in \Omega} I(X;Y),
        \end{equation}    
and the ergodic capacity with receiver CSI be defined as 
        \begin{equation}
            C_{\text{CSI}} = \sup_{F \in \Omega} \mathbb{E}_H[I(X;Y|H)] .
        \end{equation}
Then the following inequality holds
        \begin{equation}
            C \le C_{\text{CSI}}.
        \end{equation}
\end{lemma}
\begin{proof}
    Using the chain rule of mutual information and the independence of $X$ and $H$, we obtain $I(X;Y|H)=I(X;Y)+I(X;H|Y)\ge I(X;Y)$. Taking expectation over $H$ and maximizing over $F \in \Omega$ completes the proof.
\end{proof}
The quantity $C=\sup_{F\in\Omega}I(F)$ corresponds to the Shannon capacity of the averaged channel characterized by the transition probability $p_{Y|X}(y|x)=\mathbb{E}_H[p_{Y|X,H}(y|x,H)]$. Under the assumption of ergodic fading and independent channel realizations across channel uses, the classical channel coding theorem applies. Consequently, for any transmission rate $R<C$, there exists a sequence of codes whose decoding error probability approaches zero as the blocklength increases. In the absence of CSI at both the transmitter and the receiver, codebooks are generated according to the optimal input distribution $F^\star$. The receiver performs maximum-likelihood decoding with respect to the averaged channel transition probability $p_{Y|X}(y|x)$. Since the capacity-achieving distribution is discrete with a finite number of mass points, the resulting signaling scheme corresponds to a finite-alphabet amplitude modulation.

We note that the ergodic capacity considered in this work characterizes the average achievable performance over multiple independent channel realizations. While this metric provides fundamental insights into the information-energy trade-offs of SLIPT systems, particularly in scenarios where channel state information is limited or unavailable (e.g., underwater systems with constrained feedback links), outage-based capacity metrics are more suitable for ensuring user-specific reliability when CSI is available. Therefore, extending the proposed analysis to outage capacity constitutes an important direction for future work.

\vspace{-4mm}
\subsection{Information Transfer}
We consider a memoryless discrete-time lognormal channel. Let $x\in \mathbb{R}^+$ be the transmit signal. For a given $A>0$, and $\varepsilon>0$, the PP and AP constraints are $0\leq x\leq A$, and $\mathbb{E}\{x\}\leq\varepsilon$, respectively~\cite{Lapidoth_2009}. $\mathbb{E}\{\cdot\}$ denotes the expectation operator. {In OWC systems, LEDs and laser diodes operate linearly only within a limited input range; exceeding this range can cause nonlinear distortion and device degradation. To ensure reliable operation, we impose a PP constraint based on practical transmitter limitations and optical conversion efficiency, using values consistent with prior studies. Additionally, the AP constraint reflects thermal and energy efficiency considerations and aligns with safety standards for human exposure to optical radiation. The chosen AP and PP values ensure our simulation framework remains both realistic and implementable~\cite{Chaaban_2022}.} The received electrical signal at the information receiver is expressed as
    \begin{align}\label{equ:e1}
        y = a R_Ph_1x+n,
    \end{align}
where $a$ is the electrical-to-optical conversion efficiency at the LED/LD, $R_P$ is the responsivity of the PD, and $n$ is the real AWGN with zero mean and variance $\sigma_g^2$ \textit{i.e.}, $n\sim\mathcal{N}(0,\sigma_g^2)$. 

\vspace{-4mm}
\subsection{Power Transfer}
To perform EH, the PV receiver generates a current from the incident optical power and stores charges in a battery. The received current signal at the PV receiver is expressed as
    \begin{align}\label{equ:e2}
        y' = a R_Eh_2x,
    \end{align}
where $R_E$ is the responsivity of the PV cell. The effect of the noise can be neglected in the PV cell receivers~\cite{Geofeng_2019}. The harvested energy at the EH receiver for a time duration $T$ is expressed as~\cite{Ding_2018,Sait_2019}
    \begin{align}\label{equ:e3}
        E_H = f_ETI_{s}V_{oc},
    \end{align} 
where $f_E$ is the fill factor of the PV cell, $I_{s}$ is the shunt current, and $V_{oc}$ is the open circuit voltage. Within our mathematical framework, $I_{s}$ is expressed as $a R_Eh_2x$ \textit{i.e.}, the received current signal at the PV receiver, and $V_{oc}$ is  $v_t\ln\left(1+\frac{a R_E h_2 x}{I_0}\right)$, where $v_t$ is the thermal voltage, and $I_0$ is the dark saturation current~\cite{Ding_2018}. By combining these expressions and taking the expectation, the average harvest energy can be expressed as
    \begin{align}\label{equ:e4}
        \bar{E}_H = \mathbb{E}\left\{f_Ev_tT a R_Eh_2x\ln\left(1+\frac{a R_E h_2 x}{I_0}\right)\right\}.
    \end{align}
To satisfy the EH requirement at the receiver, $\bar{E}_H$ needs to be greater than the given threshold value $E_{th}$, \textit{i.e.}, $E_{th}\le \bar{E}_H$.

To capture the nonlinear behavior of the PV EH circuit, the EH model in \eqref{equ:e4} follows the widely adopted PV-based formulation (e.g., \cite{Ding_2018}), which characterizes the relationship between the photocurrent, open-circuit voltage, and harvested power. In practice, PV cells are typically operated near their optimal operating point using circuit-level techniques such as maximum power point tracking (MPPT). Importantly, MPPT relies only on local voltage and current measurements within the harvesting circuit and does not require CSI. Therefore, the adopted EH model is consistent with the no-CSI assumption and remains compatible with low-complexity receiver architectures in SLIPT systems. Investigating more detailed EH models that do not rely on MPPT and capture additional circuit-level nonlinearities (e.g., \cite{Nikita_2023}) is an interesting direction for future work.

\subsection{Problem Formulation}
We form an optimization problem to maximize the average mutual information between the channel input $X$ and the channel output $Y$ subject to PP and AP constraints at the transmitter, and a minimum harvested energy constraint at the EH receiver. The optimization problem, P1 is expressed as
	\begin{maxi!}|s|
		{F\in \mathcal{F}_A}{I(F)=\int_{0}^{A} \int_{y} p(y|x)\log_2\frac{p(y|x)}{p(y;F)}\d y \d F(x)} 
		{\label{Eqopt}}{} 
        \addConstraint{\mathbb{E}\{X\} \leq \varepsilon}  
        \addConstraint{\mathbb{E}\left\{bX\ln\left(1+cX\right)\right\}\ge E_{th}},
	\end{maxi!}
where $b = f_Ev_tTa R_Eh_{2}$ and $c = a R_Eh_{2}/I_0$ are constants, $\mathcal{F}_A$ is the set of all input distributions that satisfy the PP constraint \textit{i.e.}, $\mathcal{F}_A = \left\{F\in\mathcal{F}, \int_{0}^{A}\d F(x) = 1\right\}$, and $\mathcal{F}$ is the set of all possible input distributions. The mutual information between the random variables $X$ and $Y$ of input and output, respectively, can be written as a function of the input distribution $F$  
\begin{equation}
    I(F)\defn\int_{0}^{A}   i(x;F)   \d F(x),
\end{equation}
where $i(x;F)=\int_y p(y|x)\log_2\frac{p(y|x)}{p(y;F)}\d y$ is the marginal information density. Denote by $g_j:F \rightarrow \mathbb{R}$ with $j\in\{1,2\}$ the following functions
\begin{IEEEeqnarray}{cCl}
         g_1(F)&\defn& \int_{0}^{A} x \d F(x)-\varepsilon,\\
         g_2(F)&\defn& E_{\text{th}}- \int_{0}^{A} bx\ln(1+cx) \d F(x), 
\end{IEEEeqnarray}
and let $\Omega$ be the set of all input distributions, such that
\begin{equation}
   \hspace{-1.7mm} \Omega = \left\{F\in\mathcal{F};\int_{0}^{A}dF(x) = 1; g_j(F)\leq 0; j\in \{1,2\}\right\}.
\end{equation}
Hence, the optimization problem in \eqref{Eqopt} could be written in the compact form as $\displaystyle C = \sup_{F\in \Omega} I(F)$.

The optimization problem in \eqref{Eqopt}, while non-convex in the signal amplitude $x$, is a convex problem in the space of input distributions $F_X(x)$, since the objective is concave and the constraints are linear in $F_X(x)$. This convexity justifies the use of the KKT conditions for characterizing the properties of the optimal distribution.

\vspace{-2mm}
\section{Discreteness of the Optimal Input Distribution}\label{discreteness}
The discreteness of the optimal input distribution for maximizing mutual information in IM/DD channels has been well established in prior works, including \cite{Lapidoth_2009}, which analyzed the capacity of IM/DD optical channels under PP and/or AP constraints. These studies demonstrated that, for the AWGN channel, the capacity-achieving input distribution consists of a finite number of mass points. However, our analysis extends these results by considering a more general setting that includes lognormal fading, which introduces additional  variations in the received signal that were not accounted for in previous AWGN-based formulations. Furthermore, unlike prior studies that primarily focused on power constraints, our optimization problem incorporates a nonlinear EH constraint, represented by \( \mathbb{E}[bX \ln(1 + cX)] \geq E_{\text{th}} \). {The introduction of the nonlinear EH constraint $\mathbb{E}[bX \ln(1+CX)]$ fundamentally alters the structure of the optimization problem. Unlike linear constraints, this nonlinear condition directly impacts the feasible set of admissible input distributions. As a result, the capacity-achieving distribution no longer exhibits the same structure observed in systems with only PP or AP constraints. Our numerical results (see Figs. 5 and 6) reveal that the number and location of the distribution’s mass points change as a function of the EH threshold, confirming that the nonlinear EH constraint induces new tradeoff regimes.} Additionally, our work considers a practical SLIPT system where a PD is used for information reception, while a PV cell is used for EH, introducing new trade-offs between maximizing mutual information and ensuring minimum harvested energy. Our results establish the discreteness property of the optimal input distribution under these additional constraints and fading effects, providing new insights into the design of IM/DD-based SLIPT systems.

We study the properties of the capacity-achieving distribution of the SLIPT system, and the solution of the optimization
problem in \eqref{Eqopt}. To this end, the mathematical framework
proposed in \cite{Smith_1971} is extended. Firstly, Theorem \ref{Theorem 1} establishes the existence and the uniqueness of the optimal input distribution. Secondly, by using the Lagrangian Theorem, the dual equivalent problem is given by Corollary \ref{Corollary 1}. Thirdly, we provide necessary and sufficient conditions for the optimal input distribution in Theorem \ref{Theorem 2} and extend to a more useful set in Corollary \ref{Corollary 2}. Finally, we show that the capacity-achieving input distribution is discrete in Theorem \ref{Theorem 3}.

\begin{theorem}\label{Theorem 1}
The capacity $C$ is achieved by a unique input distribution $F^*$ \textit{i.e.},
\begin{equation}\label{equ:t1}
    C = \sup_{F\in \Omega} I(F) = I(F^*).
\end{equation}
\end{theorem}
\begin{proof}
The proof is presented in Appendix \ref{ProofOfTheorem1}.
\end{proof}


\begin{corollary}\label{Corollary 1}
 Strong duality holds for the optimization problem in~\eqref{Eqopt} i.e., there are constants, $\lambda_j\geq 0$, for $j\in\{1,2\}$ such that
\begin{align}
    C = \sup_{F\in \Omega} I(F) - \sum_{j=1}^{2} \lambda_j g_j(F).
\end{align}
\end{corollary}

\begin{proof}
The proof is presented in Appendix \ref{ProofOfCorollary1}.
\end{proof}

\begin{theorem}\label{Theorem 2}
$F^*$ is the capacity-achieving input distribution, if and only if, $\forall F\in \Omega$ there exist $\lambda_j>0$, for $j\in\{1,2\}$ such that
\begin{align}
    \int_0^A [ i(x;F^*) &- \lambda_1x + \lambda_2 bx\ln(1+cx) ]\d F(x)\nonumber \\
    &\leq C- \lambda_1\varepsilon +\lambda_2 E_{th}.
\end{align}
\end{theorem}

\begin{proof}
The proof is presented in Appendix \ref{ProofOfTheorem2}.
\end{proof}

\begin{corollary}\label{Corollary 2}
Let $\text{Supp}(F^*)$ be the points of support of a distribution $F^*$. Then, $F^*$ is the optimal input distribution, if there exist $\lambda_1\geq 0$, and $\lambda_2\geq 0$, such that
\begin{align}\label{equ:fequal}
    \lambda_1(x&-\varepsilon)-\lambda_2(bx\ln(1+cx)-E_{th})+C
    \nonumber \\
    &-\int_y p(y|x)\log_2\frac{p(y|x)}{p(y;F^*)}\d y\geq 0,
    \end{align}
for all $x$, with equality if $x\in \text{Supp}(F^*)$.
\end{corollary}

\begin{proof}
The proof is presented in Appendix \ref{ProofOfCorollary2}.
\end{proof}
In the following, we will show that the equality in Corollary \ref{Corollary 2} can not be satisfied in a set that has an accumulation point, hence the support of $F^{*}$ must be discrete. The discreteness property of the optimal input distribution is given by the following theorem,
\begin{theorem}\label{Theorem 3}
    The optimal input distribution that achieves the capacity in \eqref{Eqopt} is discrete with a finite number of mass points.
\end{theorem}
\begin{proof}
    The proof is presented in Appendix \ref{ProofOfTheorem3}.
\end{proof}

After deriving the discrete properties of the optimal input distribution, it is natural to extend the analysis to explore the behavior of the input distribution under specific channel conditions. In particular, the high SNR regime presents an interesting case where the impact of noise becomes negligible, allowing us to derive an analytical expression for an achievable input distribution. This regime is important for characterizing system performance under favorable conditions and for understanding how AP and EH constraints shape the information energy capacity region of SLIPT systems. Therefore, in the following section, we present a detailed analysis of the optimal input distribution for the lognormal channel in the high SNR regime.

\vspace{-2mm}
\section{Optimal Input Distribution for Lognormal Channel in High SNR Regime}
\label{sec:opt_input}
In this section, we derive the optimal input distribution for the lognormal fading channel under the high SNR regime, considering the PP, AP, and EH constraints. Unlike the AWGN channel, the lognormal fading introduces more complexity, which we address by directly working with the original form of the EH constraint. By following the discussion in Section \ref{discreteness}, the optimal input distribution is discrete. For simplicity, we consider a discrete distribution for the input $x$ over the equally-spaced alphabet $x\in\{0,l,2l,\ldots,(N-1)l\}$, where $0\le lk \le A, \forall k\in \{0,\ldots, N-1\}$, $l = \frac{A}{N-1}$, and $N$ is the number of mass points. The probability mass function (PMF) of $0\le p_x(k)\le 1$ is assigned to each point such that
\vspace{-2mm}
    \begin{subequations}
    \label{eqConstr}
        \begin{IEEEeqnarray}{rcl}
            \sum_{k=0}^{N-1}p_x(k) &=& 1,\\
            \sum_{k=0}^{N-1}lkp_x(k) &=& \varepsilon,\\
            \sum_{k=0}^{N}\left(blk\ln\left(1+clk\right)\right)p_x(k)&\ge& E_{th}.
      \end{IEEEeqnarray}
    \end{subequations}
Hence, $p_x(k)$ satisfies the non-negativity, PP, AP, and EH constraints.

The objective is to maximize the mutual information \( I(F) \) between the input \( X \) and output \( Y \), given by
    \begin{align}
    I(F) \approx \sum_{k=0}^{N-1} p_x(k) \log_2 \left( \frac{p(y|x_k)}{p(y; F)} \right),
    \end{align}
where \( p_x(k) \) represents the PMF at the discrete points \( l_k = kl \) for \( k \in \{0, 1, \ldots, N-1\} \), and \( l = \frac{A}{N-1} \).
The Lagrangian function incorporating the constraints in \eqref{eqConstr} is given by 
    \begin{align}\label{eqLagr}
    &\mathcal{L}(p_x(k), \lambda_1, \lambda_2, \lambda_3) \hspace{-0.5mm}= \hspace{-0.5mm}I(F)-\hspace{-0.5ex} \lambda_1 \hspace{-0.5mm}\left(\sum_{k=0}^{N-1} l_k p_x(k) - \varepsilon \right)\hspace{-0.5mm}-\\
    &\lambda_2 \left(E_{th} \hspace{-0.5ex}- \hspace{-0.5ex}\sum_{k=0}^{N-1} b l_k \ln(1 + c l_k) p_x(k)\right)
    \hspace{-0.5ex}- \hspace{-0.5ex}\lambda_3 \hspace{-0.5mm}\left(\sum_{k=0}^{N-1} p_x(k) \hspace{-0.5mm}-\hspace{-0.5mm} 1\right)\hspace{-0.5mm},\nonumber
    \end{align}
where \( \lambda_1 \), \( \lambda_2 \), and \( \lambda_3 \) are the Lagrange multipliers for the AP, EH, and normalization constraints, respectively.
Taking the derivative of the Lagrangian in \eqref{eqLagr} with respect to \( p_x(k) \) and setting it to zero gives the optimality condition
    \begin{align}
    \frac{\partial \mathcal{L}}{\partial p_x(k)} \hspace{-0.5ex}= \hspace{-0.5ex}\log_2 \left( \frac{p(y|x_k)}{p(y; F)} \right)\hspace{-0.5ex} -\hspace{-0.5ex} \lambda_1 l_k\hspace{-0.5ex} - \hspace{-0.5ex}\lambda_2 b l_k\ln(1 + c l_k) \hspace{-0.5ex}-\hspace{-0.5ex} \lambda_3 = 0.
    \end{align}
Rearranging this, we find
    \begin{align}
    p_x(k) = \frac{1}{Z} \exp \left( -\lambda_1 l_k - \lambda_2 b l_k \ln(1 + c l_k) \right),
    \end{align}
where \( Z \) is the normalization constant ensuring that \( \sum_{k=0}^{N-1} p_x(k) = 1 \), and is given by
    \begin{align}
    Z = \sum_{k=0}^{N-1} \exp \left( -\lambda_1 l_k - \lambda_2 b l_k \ln(1 + c l_k) \right).
    \end{align}
The optimal input distribution in the high SNR regime balances the contributions from the AP and EH constraints. The derived expression

\begin{align}
p_x(k) = \frac{1}{Z} \exp \left( -\lambda_1 l_k - \lambda_2 b l_k \ln(1 + c l_k) \right),
\end{align}shows that the distribution is influenced by both the linear term \( l_k \) and the logarithmic term \( \ln(1 + c l_k) \), reflecting the dual impact of the power and energy-harvesting requirements.
This analytical expression provides a practical approach to determining the achievable input distribution in high SNR scenarios, where the noise is negligible, and the system operates near its capacity limits. After characterizing  the achievable input distribution for the high SNR regime, we will focus now on the particular scenario where the binary input distribution is no longer optimal. This is a critical point in SLIPT systems, since the binary distribution maximizes both information and energy transfer simultaneously and therefore there is not a trade-off between them. In Section V, we analyze the properties of the mass points under these conditions and characterize a sufficient condition on the PP constraint $A$, where the binary distribution is no longer optimal.

\section{Properties of the mass points}\label{sec:mass_points}

We provide an analysis of the properties of the optimal input distribution \( F^{\star} \) under both the PP and AP constraints. It is well-established that the optimal input distribution is discrete, and its cumulative distribution function (cdf) is fully described by the vectors \(\bs{q}\) and \(\bs{x}\), along with the number of mass points \( N \). The positions of the mass points are defined by the vector
\[
\bs{x} = (x_1, \ldots, x_N),
\]
while the corresponding probabilities are given by
\[
\bs{q} = (q_1, \ldots, q_N).
\]
For simplicity, we assume the mass points are ordered such that \( x_1 < x_2 < \ldots < x_N < A \)~\cite{Khalfet_2023}. Furthermore, the optimal input distribution \( F^{\star}_{N} \) is uniquely determined by the triplet \((\bs{q}^{\star}, \bs{x}^{\star}, N)\). The optimality conditions for \( F^{\star}_{N} \), parameterized by non-negative constants \(\alpha_1 \geq 0\) and \(\alpha_2 \geq 0\), are expressed as
\[
i(x; F^{\star}_N) \leq C + \alpha_1 (x - \varepsilon) - \alpha_2 (\mathcal{E}(x) - E_{\mathrm{th}}), \quad \forall x \in [0, A],
\]
with equality holding at the mass points:
\[
i(x_i^{\star}; F^{\star}_N) = C + \alpha_1 (x_i^{\star} - \varepsilon) - \alpha_2 (\mathcal{E}(x_i^{\star}) - E_{\mathrm{th}}).
\]
Here, the function \(\mathcal{E}(x)\) is defined as
\[
\mathcal{E}(x) = b x \ln(1 + c x).
\]

\begin{remark}
The input distribution has a necessary mass point at zero; by contradiction, we assume that  $0<x_1^{\star}$, then we can define a vector for a new mass point localisation 
as follows
\begin{equation}
    \hat{\bs{x}^{\star}}=[\hat{x}^{\star}_1,\ldots, \hat{x}^{\star}_n],
\end{equation}
with $\hat{x}^{\star}_i=(x_i^{\star}-x_1^{\star}$). Denote by $F^{'}_{n^{\star},A, \epsilon}$ the new distribution characterized by $(\bs{q}^{\star},\bs{x}^{\star},n^{\star})$. By using the fact that the differential entropy is unchanged by the translations, we have
\begin{equation}
    H(Y;F^{'}_{n^{\star},A, \varepsilon})= H(Y;F_{n^{\star},A, \varepsilon}),
\end{equation}
which is a contradiction, since the capacity achieving distribution is unique. Therefore, $x^{\star}_1=0$.
\end{remark}

In this section, we analyze the transition points of the optimal input distribution by deriving the behavior of the mutual information \( I(X; Y) \) with respect to the amplitude \( A \) of the input signal~\cite{Shamai_2010}. The mutual information between the input \( X \) and the output \( Y \) for the lognormal fading channel can be expressed as
    \begin{equation}
    I(X; Y) = H(Y) - H(Y | X),
    \end{equation}
where \( H(Y) \) is the entropy of the output \( Y \), and \( H(Y | X) \) is the conditional entropy given \( X \). Given that the noise \( V \) is Gaussian, the conditional entropy \( H(Y | X) \) can be expressed as
    \begin{equation}
    H(Y | X) = \frac{1}{2} \log_2(2\pi e \sigma_g^2).
    \end{equation}
The output entropy \( H(Y) \) is defined as
    \begin{equation}
    H(Y) = -\int_{-\infty}^{\infty} p_Y(y) \log_2 p_Y(y) \, dy,
    \end{equation}
where \( p_Y(y) \) is the marginal distribution of \( Y \), derived from the conditional distribution \( p_{Y|X}(y|x) \) as
    \begin{equation}
    p_Y(y) = \int_{0}^{A} p_{Y|X}(y|x) p_X(x) \, dx.
    \end{equation}
To understand the behavior of the mutual information with respect to the amplitude \( A \), we differentiate \( I(X; Y) \) with respect to \( A \)
    \begin{equation}
    \frac{\partial I(X; Y)}{\partial A} = \frac{\partial H(Y)}{\partial A} - \frac{\partial H(Y | X)}{\partial A}.
\end{equation}
Since \( H(Y | X) \) does not depend on \( A \), we have
    \begin{equation}
    \frac{\partial H(Y | X)}{\partial A} = 0.
    \end{equation}
Thus, the derivative simplifies to
    \begin{equation}
    \frac{\partial I(X; Y)}{\partial A} = \frac{\partial H(Y)}{\partial A}.
    \end{equation}
The entropy \( H(Y) \) can be written as
    \begin{align}
    H(Y) &= -\int_{-\infty}^{\infty} \left( \int_{0}^{A} p_{Y|X}(y|x) p_X(x) \, dx \right) \nonumber \\
    &\times\log_2 \left( \int_{0}^{A} p_{Y|X}(y|x) p_X(x) \, dx \right) dy.
    \end{align}
To differentiate this with respect to \( A \), we first consider the inner integral
    \begin{equation}
    \frac{\partial}{\partial A} \left( \int_{0}^{A} p_{Y|X}(y|x) p_X(x) \, dx \right) = p_{Y|X}(y|A) p_X(A).
    \end{equation}
Thus, the derivative of \( H(Y) \) with respect to \( A \) becomes
    \begin{IEEEeqnarray}{l}
    \frac{\partial H(Y)}{\partial A} = -\int_{-\infty}^{\infty}  p_{Y|X}(y|A) p_X(A)\\
    \times \log_2 \left( \int_{0}^{A} p_{Y|X}(y|x) p_X(x)  dx \right) + p_{Y|X}(y|A) p_X(A)  dy.\nonumber
    \end{IEEEeqnarray}
We now consider the specific input distribution \( F_{3,A,\varepsilon} \), which has three mass points, described as
    \begin{equation}
    F_{3,A,\varepsilon}(x) = 
    \begin{cases} 
    1 - \frac{\varepsilon}{A} - \frac{x_1}{A}q, & x = 0 \quad \text{(i.e., } q_0\text{)} \\
    q, & x = x_1 \quad \text{(i.e., } q_1\text{)} \\
    \frac{\varepsilon + x_1 q}{A}, & x = A \quad \text{(i.e., } q_2\text{)}
    \end{cases}.
    \end{equation}
Substituting the expression for \( q_2 = \frac{\varepsilon + x_1 q}{A} \) into the derivative, we get
    \begin{equation}
        \frac{\partial I(F_{3,A,\varepsilon})}{\partial A} \hspace{-0.5ex}= \hspace{-0.5ex}\frac{\varepsilon + x_1 q}{A}\hspace{-0.8ex} \int_{-\infty}^{\infty} \hspace{-0.5ex}\frac{\partial p_{Y|X}(y|A)}{\partial A} \log_2 \hspace{-0.5ex} \left( \frac{p_{Y|X}(y|A)}{p_Y(y)} \right) \hspace{-0.5ex}dy.
    \end{equation}

The transition probability \( p_{Y|X}(y|x) \) for the lognormal channel is given by
    \begin{align}
        p_{Y|X}(y|x) &=
        \frac{e^{-y^2/2\sigma_g^2}}{4\pi\sigma_g\sigma_{X_l}}
        \sum_{n=0}^{\infty} \frac{1}{n!} H_n\left(\frac{y}{\sigma_g}\right) \left(\frac{aR_P h_{1,l} x}{\sigma_g}\right)^n \nonumber \\
        &\times\exp\left(\frac{\mu_{X_l} + 2n\sigma_{X_l}^2}{4\sigma_{X_l}^2}\right).
    \end{align}
By taking the derivative of \( p_{Y|X}(y|A) \) with respect to \( A \), we have
    \begin{align}
        \frac{\partial p_{Y|X}(y|A)}{\partial A} &= \frac{e^{-y^2/2\sigma_g^2}}{4\pi\sigma_g\sigma_{X_l}} \sum_{n=0}^{\infty} \frac{n}{n!} H_n\left(\frac{y}{\sigma_g}\right) \left(\frac{aR_P h_{1,l}}{\sigma_g}\right)^n \nonumber \\ 
        &\times A^{n-1} \exp\left(\frac{\mu_{X_l} + 2n\sigma_{X_l}^2}{4\sigma_{X_l}^2}\right).
    \end{align}

\subsection{Final Expression for the Derivative}

Substituting the derivative into the expression for \( \frac{\partial I(F_{3,A,\varepsilon})}{\partial A} \), we obtain
    \begin{align}
        &\frac{\partial I(F_{3,A,\varepsilon})}{\partial A} \hspace{-0.5ex}= \hspace{-0.5ex}\frac{\varepsilon + x_1 q}{A} \hspace{-0.5ex} \int_{-\infty}^{\infty} \frac{e^{-y^2/2\sigma_g^2}}{4\pi\sigma_g\sigma_{X_l}} \sum_{n=0}^{\infty} \frac{n}{n!} H_n \left(\frac{y}{\sigma_g}\right) A^{n-1} \nonumber \\
        &\times \hspace{-0.5ex} \left(\frac{aRP h_{1,l}}{\sigma_g}\right)^n \hspace{-0.5ex}\exp\left(\frac{\mu_{X_l} + 2n\sigma_{X_l}^2}{4\sigma_{X_l}^2}\right) \hspace{-0.5ex}\log_2 \left( \frac{p_{Y|X}(y|A)}{p_Y(y)} \right) dy.
    \end{align}
This expression represents the rate of change of mutual information with respect to the amplitude \( A \), considering the specific distribution \( F_3 \). Hence, by solving $ \frac{\partial I(F_{3,A,\varepsilon})}{\partial A} >0$, we obtain a  sufficient condition for the transition point ${A}$, where the binary distribution is no longer optimal.
Following the detailed examination of the lognormal fading channel and the associated transition points (in Section \ref{sec:mass_points}), we now shift our focus to the evaluation of the CIECL framework. Section \ref{sec:CIECL} presents the development and application of this framework, leveraging deep learning techniques to optimize input distributions and estimate the capacity region under varying system conditions.

\vspace{-3mm}
\section{Cooperative Information-Energy Capacity Learning}\label{sec:CIECL}
    \begin{figure}[!t]
        \centering
        \includegraphics[width=0.85\columnwidth]{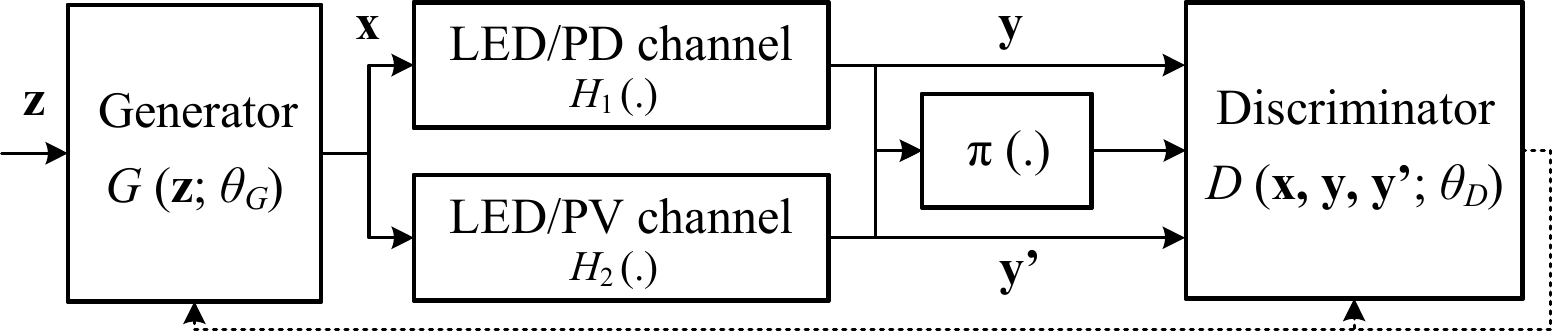}
        \caption{Proposed information-energy capacity learning framework inspired by GAN.}
        \label{fig:GAN}
        \vspace{-5mm}
    \end{figure}
        \begin{figure*}[!t]
        \centering
        \includegraphics[width=0.99\textwidth]{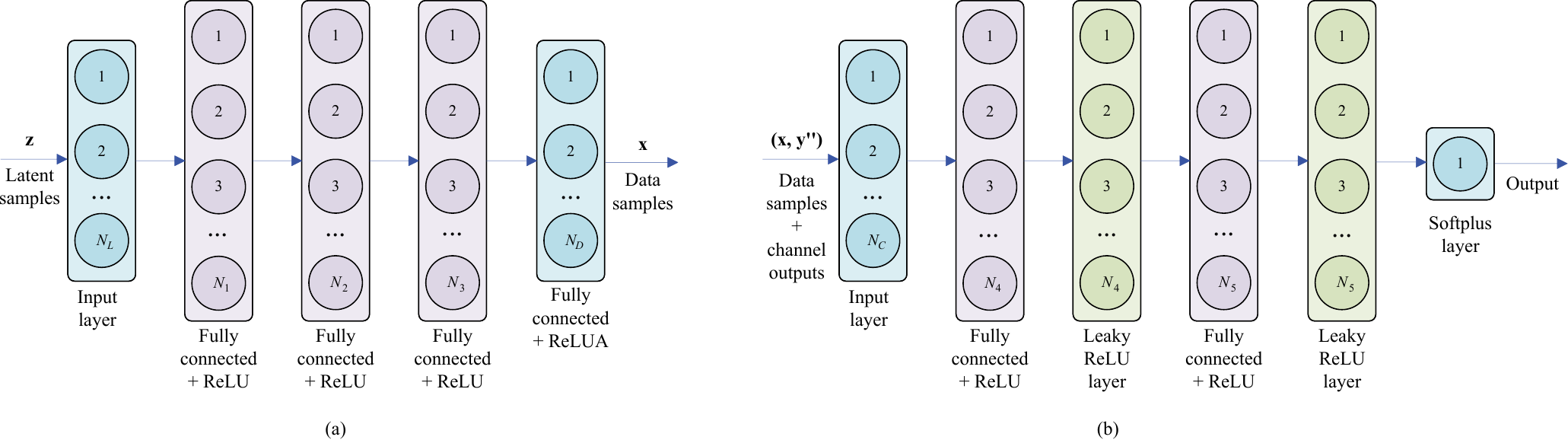}
        \caption{Neural networks used in CIECL. (a) Generator. (b) Discriminator.}
        \label{fig:ANN}
        \vspace{-5mm}
    \end{figure*}
In this section, we present the CIECL approach to learn the optimum input distribution and the information-capacity region of the SLIPT system. To this end, we modify the CORTICAL framework presented in\cite{Letizia_2023} to account for the SLIPT system with lognormal channel, and the energy harvesting constraint. According to this framework, a GAN is adopted~\cite{Antonia_2018}. In the following, we explain the operation principle of GAN.

In the GAN framework, two neural networks named, generator ($G$) and discriminator ($D$) are used as shown in Fig. \ref{fig:GAN}. In the adversarial training process of $G$, the probability of $D$ making a mistake is maximized. Denote $\mathbf{x}\sim p_{\text{data}}(\mathbf{x})$ are the data samples and $\mathbf{z}\sim p_{Z}(\mathbf{z})$ are the latent samples. The Nash equilibrium is reached when the minimization over $G$ and the maximization over $D$ of the value function given by
    \begin{align}
        \mathcal{V}(G,D) &= \mathbb{E}_{\mathbf{x}\sim p_{\text{data}}(\mathbf{x})} \Big[\log(D(\mathbf{x}))\Big] \nonumber \\
        &+ \mathbb{E}_{\mathbf{z}\sim p_{Z}(\mathbf{z})} \Big[\log(1-D(G(\mathbf{z})))\Big],
    \end{align}
is obtained so that $G(\mathbf{z})\sim p_{\text{data}}(\mathbf{x})$, where $\mathbb{E}_{\mathbf{x}\sim p_{\text{data}}(\mathbf{x})}$ is the expectation over $p_{\text{data}}(\mathbf{x})$, and $\mathbb{E}_{\mathbf{z}\sim p_{Z}(\mathbf{z})}$ is the expectation over $p_{Z}(\mathbf{z})$, respectively. The information-energy capacity region estimation problem requires $G$ to learn the input distribution maximizing the mutual information. Hence, $G$ and $D$ need to play a cooperative max-max game with the objective for $G$ to produce channel input samples and the objective of $D$ is to distinguish paired and unpaired channel input-output samples. The discriminator uses the output samples from the generator $\mathbf{x}$, and channels' outputs samples $\mathbf{y}$ and $\mathbf{y'}$ to distinguish paired and unpaired samples. 

The proposed cooperative framework that learns both generator and discriminator is shown in Fig. \ref{fig:GAN}, where $H_1$ and $H_2$ represent the channels between the LED transmitter and the PD and PV, respectively, and $\pi(\cdot)$ is a permutation function. {The generator network \( G(z) \) generates input samples that approximate the optimal input distribution, while the discriminator network \( D(x, y, y') \) estimates mutual information by distinguishing between true input-output samples and randomly paired samples. The overall objective is to refine \( G(z) \) such that the generated inputs maximize mutual information while satisfies the AP and EH constraints. To formalize this, we define the loss function as
   \begin{align}
        &J_\alpha(G, D) = \alpha \mathbb{E}_{z \sim p_Z} \left[ \log D(G(z), H_1(G(z)), H_2(G(z))) \right] \nonumber \\
        &\quad - \mathbb{E}_{z \sim p_Z} \left[ \log D(G(z), \pi(H_1(G(z))), \pi(H_2(G(z)))) \right],
    \end{align}
\noindent where \( H_1(G(z)) \) and \( H_2(G(z)) \) represent the received signals at the PD and PV cell, respectively. \( \pi(\cdot) \) is a permutation function that shuffles input-output pairs to create non-optimal sample distributions for training the discriminator. In addition to mutual information maximization, we introduce penalty terms to enforce SLIPT constraints
     \begin{align}
        &J_{\text{total}}(G, D) = J_\alpha(G, D) - \lambda_A \max (||G(z)||_1 - A, 0) \nonumber \\
        &- \lambda_P \max (\mathbb{E}[||G(z)||_1] - \varepsilon, 0) \nonumber \\
        &- \lambda_E \max (E_{\text{th}} - \mathbb{E}[bG(z) \ln(1 + cG(z))], 0),
    \end{align}
\noindent where \( \lambda_A, \lambda_P, \lambda_E \) are non-negative Lagrange multipliers enforcing the PP, AP, and EH constraints, respectively. The term \( \max (||G(z)||_1 - A, 0) \) ensures that the generated inputs satisfy the peak power constraint. The term \( \max (\mathbb{E}[||G(z)||_1] - \varepsilon, 0) \) enforces the average power constraint. The term \( \max (E_{\text{th}} - \mathbb{E}[bG(z) \ln(1 + cG(z))], 0) \) guarantees that the minimum energy harvesting requirement is satisfied.}

Following the steps in~\cite{Letizia_2023}, the complete value function for the proposed SLIPT system can be expressed as 
    \begin{align}
        &\mathcal{J}_{\text{total}}(G,\hspace{-0.5mm}D) \hspace{-0.5ex}= \hspace{-0.5ex}\alpha\mathbb{E}_{\mathbf{z}\sim p_{Z}(\mathbf{z})} \hspace{-0.5mm}\Big[\hspace{-0.5mm}\log(D(G(\mathbf{z}),\hspace{-0.5mm}H_1(G(\mathbf{z})),\hspace{-0.5mm}H_2(G(\mathbf{z}))))\Big]\nonumber\\
        &- \mathbb{E}_{\mathbf{z}\sim p_{Z}(\mathbf{z})} \Big[\log(D(G(\mathbf{z})),\pi (H_1(G(\mathbf{z}))),\pi (H_2(G(\mathbf{z}))))\Big]\nonumber\\
        &-\lambda_{A}\max{(||G(\mathbf{z})||_1-A,0)}-\lambda_{P}\max(\mathbb{E}[||G(\mathbf{z)}||_1]-\varepsilon,0)\nonumber \\
        &-\lambda_E\max(E_{th}-\mathbb{E}[||bG(\mathbf{z})\ln{1+cG(\mathbf{z})}||_1],0),
    \end{align}
where $\lambda_A$, $\lambda_P$, and $\lambda_E$ equal to 0 or 1, depending on which constraints are active, $\alpha>0$ is a constant, and $||\cdot||_1$ is the $L^1$ norm.

    \begin{algorithm}[!t]
     \caption{CIECL Algorithm}
     \begin{algorithmic}[1]\label{alg:a1}
     \renewcommand{\algorithmicrequire}{\textbf{Input:}}
     \renewcommand{\algorithmicensure}{\textbf{Output:}}
     \REQUIRE $N$ training steps, $K$ discriminator steps, $\alpha$.
     \ENSURE $G_{\theta_G}$, $D_{\theta_D}$.
     \FOR {$n = 1$ to $N$}
        \FOR {$k = 1$ to $K$}
            \STATE Sample batch of $m$ noise samples $\{\mathbf{z}^{(1)},\ldots,\mathbf{z}^{(m)}\}$ from $p_Z(\mathbf{z})$;\\
            \STATE Produce batch of $m$ channel input/output triplet samples $\{(\mathbf{x}^{(1)},\mathbf{y}^{(1)},\mathbf{y}'^{(1)}),\ldots,(\mathbf{x}^{(m)},\mathbf{y}^{(m)},\mathbf{y}'^{(m)})\}$ using the generator $G_{\theta_G}$, and channel models for LED/PD and LED/PV links;\\
            \STATE Shuffle $[\mathbf{y}, \mathbf{y}']$ pairs using $\pi(.)$ function and get input/output triplet samples $\{(\mathbf{x}^{(1)},\hat{\mathbf{y}}^{(1)},\hat{\mathbf{y}}'^{(1)}),\ldots,(\mathbf{x}^{(m)},\hat{\mathbf{y}}^{(m)},\hat{\mathbf{y}}'^{(m)})\}$;\\
            \STATE Update the discriminator by ascending its stochastic gradient:
            \vspace{-3mm}
            \begin{align}
                &\nabla \theta_D \frac{1}{m}\sum_{i=1}^{m} \alpha \log \left(D_{\theta_D}(\mathbf{x}^{(i)},\mathbf{y}^{(i)},\mathbf{y}'^{(i)})\right) \nonumber \\
                &- D_{\theta_D}(\mathbf{x}^{(i)},\hat{\mathbf{y}}^{(i)},\hat{\mathbf{y}}'^{(i)}).
            \end{align}
            \vspace{-3mm}
        \ENDFOR
        \STATE Sample batch of $m$ noise samples $\{\mathbf{z}^{(1)},\ldots,\mathbf{z}^{(m)}\}$ from $p_Z(\mathbf{z})$;
        \STATE Update the generator by ascending its stochastic gradient:
            \begin{align}
                &\nabla \theta_G \frac{1}{m}\sum_{i=1}^{m} \alpha \log \bigg(D_{\theta_D}(G_{\theta_G}(\mathbf{z}^{(i)}),H_1(G_{\theta_G}(\mathbf{z}^{(i)})),\nonumber \\
                &H_2(G_{\theta_G}(\mathbf{z}^{(i)}))\bigg)
                - D_{\theta_D}(G_{\theta_G}(\mathbf{z}^{(i)}),\pi(H_1(G_{\theta_G}(\mathbf{z}^{(i)}))), \nonumber \\
                &\pi(H_2(G_{\theta_G}(\mathbf{z}^{(i)})))).
            \end{align}
     \ENDFOR
     \RETURN $G_{\theta_G}$, $D_{\theta_D}$.
     \end{algorithmic} 
     \end{algorithm} 
\subsubsection{Training Strategy of the GAN}
    
The alternating training strategy in Algorithm \ref{alg:a1} results in optimal parameters for $G$ and $D$ under the assumption of having enough capacity. In general, it is reasonable to model $G$ and $D$ with neural networks $(G,D) = (G_{\theta_G}, D_{\theta_D})$ and optimize over their parameters. We assume that the distribution of the source $p_z(\mathbf{z})$ is in a multivariate normal distribution with independent components. The function $\pi(\cdot)$ is a derangement of the batch $\mathbf{y}'' = [\mathbf{y}, \mathbf{y}']$ and it is used to obtain unpaired samples. We use $K=10$ discriminator training steps for every generator training step. The neural networks for $G$ and $D$ are shown in Fig. \ref{fig:ANN}.
We use a neural network which includes three hidden layers and an output layer for the generator. The hidden layers are fully connected, and the activation function is rectified linear unit (ReLU). The output layer is a fully connected layer in which the activation function is a modified version of the ReLU function to account the PP constraint, and can be expressed as $ReLUA(x) = \min\left(\max\left(0,x\right),A\right)$~\cite{Dubey_2022}. 
Further, the PP, AP, and EH constraints can be met by the design of $G$ using a batch normalization layer as well as setting the training loss function adequately. The discriminator is a neural network that includes two hidden layers each followed by a leaky-ReLU layer and the output layer consists of a fully connected layer with softplus activation function~\cite{Dubey_2022}.

\subsubsection{{Complexity of the GAN}}

{To evaluate the feasibility of applying the proposed GAN-based framework in real-world SLIPT systems, we analyze its computational complexity in both the training and the inference phases. The training phase consists of iterative updates for both the generator network and the discriminator network, which are fully connected neural networks with three and two hidden layers, respectively. Let \( N \) denote the number of training samples per iteration, \( L_G \) and \( L_D \) be the number of layers in the generator and discriminator, respectively, and \( D \) represent the total number of neurons per layer. The per-iteration complexity of training the generator and discriminator is approximately as
    \begin{equation}
    O(N \times D \times L_G) + O(N\times D\times L_D) = O(N \times D \times (L_G + L_D)).
    \end{equation}
During training, the generator learns to optimize the input distribution, while ensuring mutual information maximization and satisfying the AP and the EH constraints. This adversarial optimization procedure requires multiple iterations until convergence, with the total complexity scaling as
    \begin{equation}
    O(E \times N \times D \times (L_G + L_D)),
    \end{equation}
where \( E \) denotes the number of training epochs. The computational cost is manageable on modern GPUs, and once trained, the generator can produce optimal input distributions instantaneously without additional iterative optimization. 
In the inference phase, the trained generator directly outputs an optimized input distribution in constant time complexity, i.e.,
    \begin{equation}
    O(D \times L_G).
    \end{equation}
By contrast, the BA algorithm operates over a discretized input space of $M$ points. Its per iteration complexity is $\mathcal{O}(M^2)$, and the total complexity becomes $\mathcal{O}(I M^2)$ for $I$ iterations. For high-accuracy approximation or fine-grained input spaces, $M$ must be large, significantly increasing both runtime and memory usage. Furthermore, incorporating nonlinear EH constraints into BA requires feasibility checks that add additional computational complexity. Our numerical results confirm that the GAN-based approach not only offers improved mutual information performance but also reduces complexity in constrained SLIPT settings.}

\vspace{-3mm}
\section{Numerical Results and Discussion}\label{Results}

    \begin{figure}[!t]
    \centering
    \includegraphics[width=0.8\linewidth]{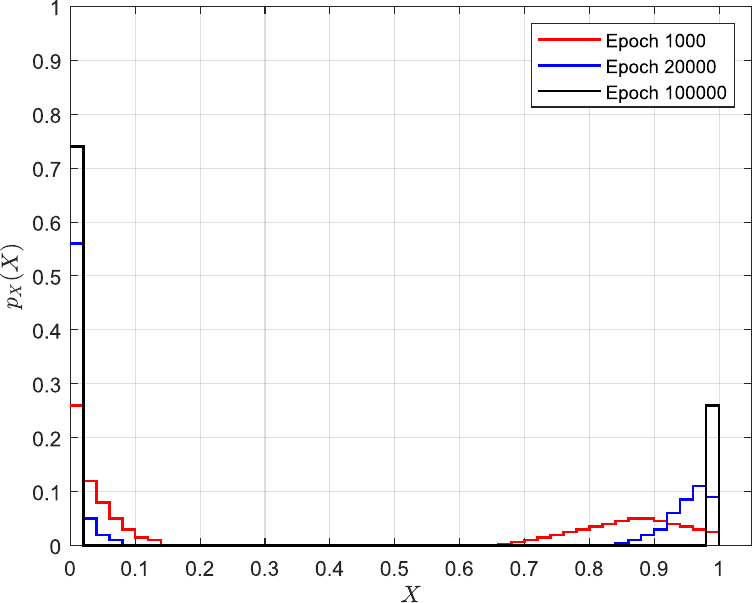}
    \caption{Optimal input distribution learned by CORTICAL
    at different training steps under lognormal channels with $A = 1$, $\varepsilon =1$, $E_{th} = 1$ mJ.} \label{figOI1}
    \vspace{-3mm}
    \end{figure}

In this section, we numerically evaluate the information-energy capacity region and the corresponding optimal input distribution. Unless stated otherwise, the parameter values are set to $a = 20$ W/A, $R_P = 0.5$ A/W, $R_E = 0.75$ A/W, $f_E = 0.5$, $T = 1$ s, $v_t = 25$ mV, $I_0 = 10^{-9}$ A, $\eta_t = \eta_r = 1$, $c(\lambda) = 0.03$ $\text{m}^{-1}$, $l = 10$ m, $\theta_{0} = 10^{\circ}$, $\theta_{{PD}} = 0^{\circ}$, $\theta_{{PV}} = 5^{\circ}$, $A_{{PD}} = 0.001$ $\text{m}^2$, $A_{{PV}} = 0.01$ $\text{m}^2$, $E_{th} = 1$ mJ, $\sigma_g^2 = 10^{-12}$, and $\sigma_{X_t}^2 = 0.1$~\cite{Ding_2018, Zhang_2019, Sait_2019}.

Fig. \ref{figOI1} shows the optimal input distribution obtained using the proposed GAN architecture for a setup with $A =1$, $\varepsilon =1$, and $E_{th} = 1$ mJ. The results are illustrated for different numbers of training epochs. The optimal input distribution is binary for the considered scenario where the PP level is very low. As an obvious result, the accuracy of the optimal input distribution increases with the number of training epochs. {In our work, we mitigate overfitting through several techniques, including 1) using small networks, 2) implementing validation, and 3) applying batch normalization. While 100,000 epochs is set as the maximum allowable training duration, the actual number of training iterations is dynamically determined based on an early stopping criterion, which prevents unnecessary training beyond the point of convergence.}

    \begin{figure}[!t]
    \centering
    \includegraphics[width=0.8\linewidth]{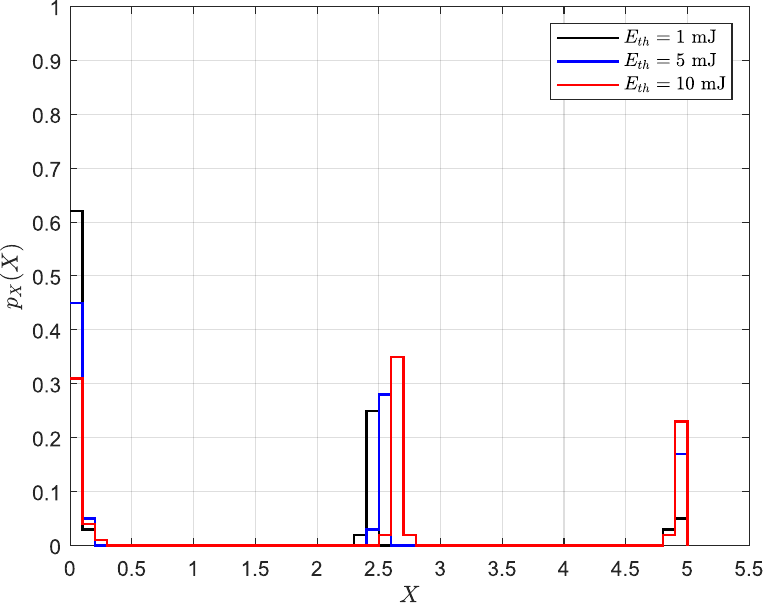}
    \caption{Optimal input distribution learned by CORTICAL at different $E_{th}$ values. $A= 5$, $\varepsilon = 2.5$, and the training epoch is $100000$ in all cases.} \label{figOI2}
        \vspace{-5mm}
    \end{figure}
    
Fig. \ref{figOI2} shows the optimal input distribution obtained using the proposed GAN architecture for a setup with $A =5$, $\varepsilon = 2.5$, and $10^5$ epochs. Compared to Fig. \ref{figOI1}, the PP constraint is increased and, as a result, the number of mass points has been increased to three. The results show that the change of the EH constraint ($E_{th}$) significantly affects the optimal input distribution. Specifically, a low $E_{th}$ value makes the distribution biased towards zero and thus the distribution is close to an exponential distribution. Moreover, a high $E_{th}$ results in biasing the distribution towards higher values of $X$ and thus distribution deviates from the exponential distribution to increase the amount of harvest energy. In particular, the mass point around $X = 2.5$ further moves towards higher values when $E_{th}$ is increased.

    \begin{figure}[t]
    \centering
    \includegraphics[width=0.78\linewidth]{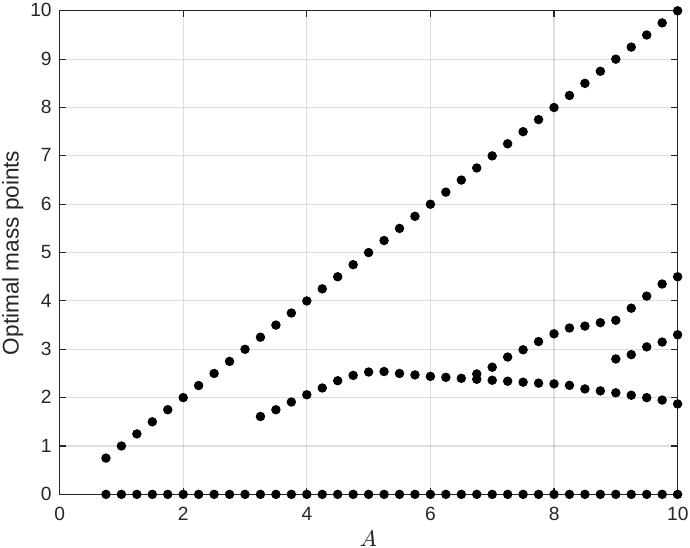}
    \vspace{-0.1cm}
    \caption{{The optimal mass points versus PP constraint; $\varepsilon = 2.5$, $E_{th}= 1 $mJ.}} \label{figmass}
    \vspace{-6mm}
    \end{figure}
{Fig. \ref{figmass} illustrates the optimal mass points obtained for various PP constraints, $A$. At low values of $A$, the problem becomes infeasible as the EH constraint cannot be satisfied. The number of mass points increases from 2 to 3, 3 to 4, and 4 to 5 at $A = 3.5, 6.5, \text{and}, 9$, respectively. However, at higher values of $A$, the AP constraint causes the mass points to concentrate at lower values.}
    
    \begin{figure}[t]
    \centering
    \includegraphics[width=0.9\linewidth]{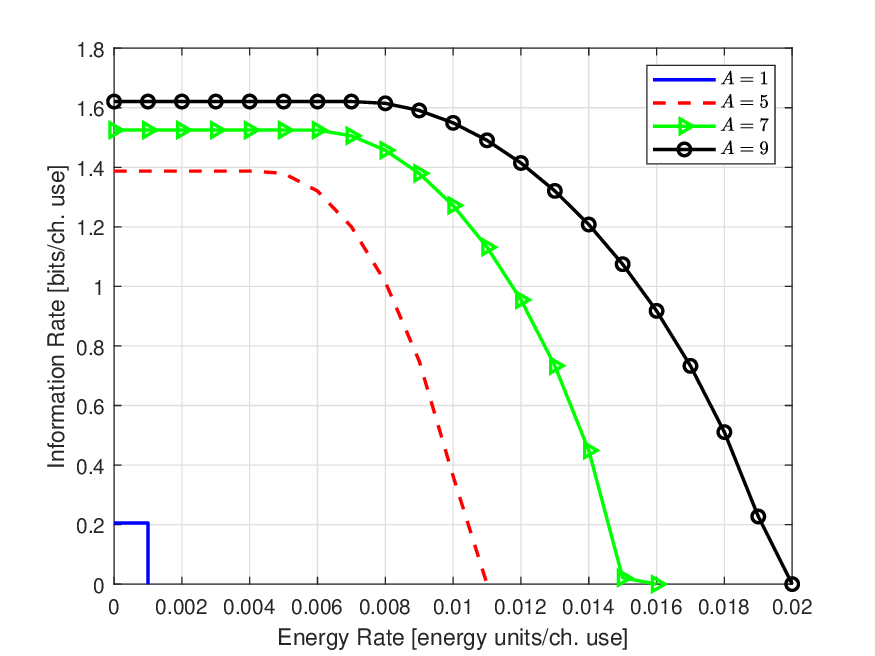}
    \vspace{-0.1cm}
    \caption{Information-energy capacity region with different PP constraints; $P=30$ dBW.} \label{figPP}
    \vspace{-6mm}
    \end{figure}

Fig. \ref{figPP} shows the information-energy capacity region for different PP constraints. The region corresponding to this scenario is determined through the resolution of the optimization problem presented in \eqref{Eqopt}. A notable trade-off emerges between the information rate and the energy rate. This trade-off becomes apparent as higher EH constraints lead the transmitter to choose a symbol with greater amplitude, consequently reducing the performance of information transfer. It is noteworthy that, for lower amplitude constraints, there is no trade-off between the two objectives. In this regime, the optimal input distribution is binary, maximizing both information and energy transfer simultaneously. 
    
Fig. \ref{figcomp} highlights the impact of the channel on the information-energy capacity region. The comparison reveals a noticeable gap between the two regions, and this discrepancy is primarily attributed to the adverse effects of the lognormal channel. Specifically, we observed that the information-energy capacity region of the Gaussian channel outperforms that of the lognormal channel. This gap in performance is due to the inherent characteristics of the lognormal channel, which introduces signal attenuation and fluctuations.
    
    \begin{figure}[t]
    \vspace{-3mm}
    \centering
    \includegraphics[width=0.9\linewidth]{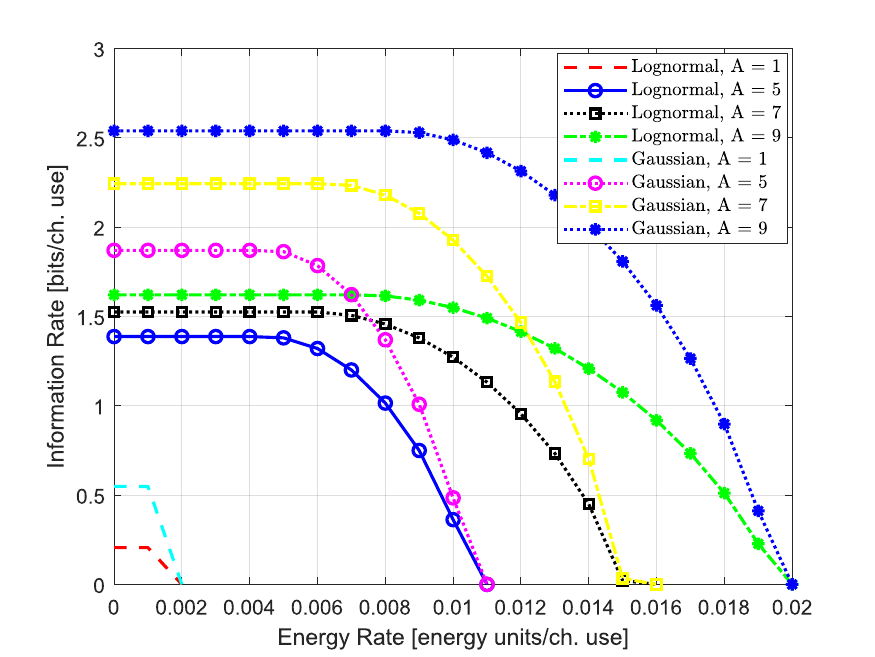}
    \vspace{-0.1cm}
    \caption{{Effect of the channel on the information-energy capacity region; $A=\{1,5,7,9\}$, $P=30$ dBW.}} \label{figcomp}
    \vspace{-5mm}
    \end{figure}

    \begin{figure}[t]
    \centering
    \includegraphics[width=0.8\linewidth]{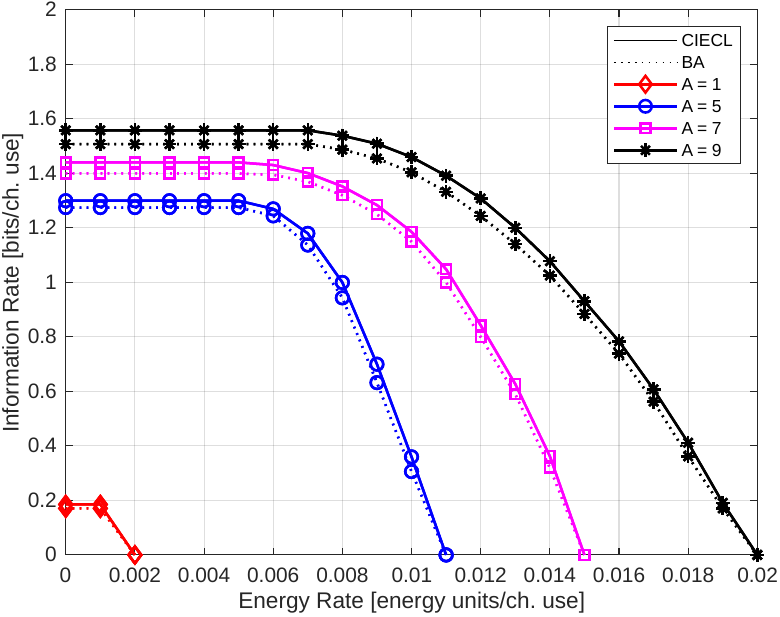}
    \vspace{-0.1cm}
    \caption{{Performance comparison of GAN and BA algorithms for capacity optimization in constrained SLIPT channels.}} \label{figGAN}
    \vspace{-5mm}
    \end{figure}
    
{An important observation in Fig. 7 and Fig. 8 is that, for low EH constraints, the energy rate increases without affecting the information rate. This observation occurs because, in this regime, the EH constraint is inactive, meaning that the optimal input distribution remains dictated solely by the mutual information maximization criterion. As a result, additional power allocation increases the harvested energy without forcing a trade-off with information transfer.
In SLIPT systems, the EH constraint is only active when the harvested energy falls below a specified threshold \( E_{\text{th}} \). When this constraint is not active, the optimal input distribution is determined independently of EH considerations, thus there is no trade off between information and power transfer in this regime. However, beyond a certain EH threshold, the system enters a regime where the EH constraint becomes active, meaning that further power increases must be allocated in a way that ensures sufficient EH. This modifies the optimal input distribution, shifting it away from the mutual information transfer optimal solution. As a result, a trade-off appears, where maximizing harvested energy leads to a corresponding reduction in information rate.} {To evaluate the effectiveness of our GAN-based approach, we compare it to the classical BA algorithm under identical PP, AP, and EH constraints. As shown in Fig. \ref{figGAN}, the GAN method consistently outperforms BA in achievable mutual information, particularly in high-EH or high-SNR regimes. This performance gap is attributed to the GAN’s ability to efficiently explore the constrained input space and adapt to nonlinearities introduced by the EH model. Furthermore, GAN exhibits faster convergence and better scalability in terms of runtime complexity, making it well-suited for real-time or adaptive SLIPT implementations.}

\vspace{-5mm}
\subsection{Practical Implications and Design Guidelines} The theoretical and numerical results presented in this work yield several key insights and practical guidelines for the design of future SLIPT systems. The proof that the capacity-achieving distribution is discrete with a finite number of mass points (Theorem 3) provides a fundamental design rule. It indicates that transmitters should not use continuous Gaussian-like signals but rather simple, finite-alphabet schemes. This finding significantly reduces the complexity of the signal design problem, guiding engineers to optimize a limited set of amplitude levels and their probabilities. Our analysis identifies clear transition points where the optimal input distribution changes structure (e.g., from binary to ternary, as shown in Fig. 6). This leads to a concrete adaptive modulation strategy. Specifically, When the peak-power constraint $A$ or the EH requirement is low, a simple binary modulation (e.g., OOK) is optimal. This is highly desirable for low-complexity, cost-sensitive devices like IoT sensors. As the power budget or energy demand increases, the optimal distribution acquires more mass points. This mandates a shift to multi-level amplitude modulation to efficiently balance the information-energy trade-off. System designers can use our results (e.g., Fig. 6) to pre-determine these transition points for their specific hardware constraints.

\vspace{-5mm}
\section{Conclusion}\label{conclusion}
In this paper, we presented a comprehensive analysis of the information-energy capacity region for SLIPT systems over lognormal-fading channels where nonlinear EH, and PP/AP constraints were considered. We derived novel expressions for the conditional probability distribution for the lognormal channel. By using Smith's framework and appropriate Hermite polynomial bases, we proved that the optimal input distribution is discrete with a finite number of mass points. Building upon this, we went beyond  characterizing the properties of the optimal input distribution by deriving an analytical expression for an achievable input distribution in the high SNR regime for the lognormal channel. This result offers a deeper understanding of SLIPT system performance in high SNR conditions, where the effects of noise are negligible. Additionally, we analyzed the optimal input distribution under low PP constraints, identifying a necessary condition on the transition point where the binary input distribution is no longer optimal. This transition marks an important trade-off between information and energy transfer in SLIPT systems. Finally, we introduced the CIECL framework, employing deep learning techniques (GANs) to generate the optimal  input distribution and estimate the information-energy capacity region for the proposed SLIPT system. Numerical results of the information-energy capacity region were presented and compared with conventional AWGN channels. Results reveal that lognormal fading significantly reduces the information-energy capacity region as compared to conventional AWGN channel conditions. Our results are helpful to identify the optimal input distributions under different PP, AP, and EH constraints. At low PP regime, the optimal input distribution is binary and no trade-off between the information rate and the EH is observed. In addition, the optimal input distribution is biased towards the zero under a low EH constraint, while it spreads across the entire signal range under a high EH constraint.

\vspace{-2mm}
\begin{appendices}
 \section{Proof of Theorem 1}
 \label{ProofOfTheorem1}
 Despite a great deal of similarity with the proof in \cite{Smith_1971,Morsi}, a sketch of the proof is presented for the sake of completeness. By following similar steps as \cite{Smith_1971}, the following properties hold:
\begin{enumerate}
    \item Continuity of the Mutual Information Functional: The mutual information functional \( I(F) \) is a weakly continuous function. That is, for any sequence of probability distributions \( \{F^{(n)}\} \) that converges weakly to \( F \), the corresponding sequence of mutual information values \( I(F^{(n)}) \) converges to \( I(F) \).
    \item Strict Concavity of the Mutual Information Functional: The mutual information functional \( I(F) \) is strictly concave over the set of feasible distributions \( \Omega \). This concavity ensures that there is a unique maximum point.
    \item Compactness of the Feasible Set \( \Omega \): The set \( \Omega \) is compact under the weak topology, which implies that an optimal distribution \( F^* \) exists.
\end{enumerate}
By the extreme value theorem, since \( I(F) \) is continuous and \( \Omega \) is compact, \( I(F) \) attains its maximum on \( \Omega \). The strict concavity of \( I(F) \) further ensures that this maximum is unique. Thus, there exists a unique distribution \( F^* \in \Omega \) that maximizes \( I(F) \), achieving the capacity \( C \). This completes the proof.

\vspace{-5mm}
\section{Proof of Corollary 1}\label{ProofOfCorollary1}
We show that strong duality holds for the optimization problem in (11), expressed as \( C = \sup_{F \in \Omega} I(F) \), subject to the constraints \( g_1(F) \leq 0 \) and \( g_2(F) \leq 0 \), where \( g_1(F) = \int_0^A x \, dF(x) - \varepsilon \) and \( g_2(F) = \int_0^A b x \ln(1 + c x) \, dF(x) - E_{th} \). To establish strong duality, we first define the Lagrangian function \( \mathcal{L}(F, \lambda_1, \lambda_2) = I(F) - \lambda_1 g_1(F) - \lambda_2 g_2(F) \), where \( \lambda_1 \geq 0 \) and \( \lambda_2 \geq 0 \) are Lagrange multipliers, and the dual function \( D(\lambda_1, \lambda_2) = \sup_{F \in \Omega} \mathcal{L}(F, \lambda_1, \lambda_2) \). Next, we verify Slater’s condition by constructing a feasible distribution \( F_0 \in \Omega \) concentrated at \( x_0 \in [0, A] \) such that \( g_1(F_0) = x_0 - \varepsilon < 0 \) and \( g_2(F_0) = b x_0 \ln(1 + c x_0) - E_{th} < 0 \). Such \( x_0 \) exists due to the linear and logarithmic nature of the constraints. Since Slater’s condition holds, the duality gap is zero, and strong duality follows: \( C = \inf_{\lambda_1 \geq 0, \lambda_2 \geq 0} D(\lambda_1, \lambda_2) = \sup_{F \in \Omega} \mathcal{L}(F, \lambda_1, \lambda_2) \). Thus, there exist non-negative Lagrange multipliers \( \lambda_1 \geq 0 \) and \( \lambda_2 \geq 0 \) such that strong duality holds for the optimization problem (7). For further details, please refer to \cite{Smith_1971}.
\section{Proof of Theorem 2}
\label{ProofOfTheorem2}
The proof relies on the necessary and sufficient conditions for optimality in constrained optimization problems. The problem involves maximizing the mutual information functional \( I(F) \) over the set \( \Omega \), subject to the constraints defined in equations (11a), (11b), and (11c). We begin by defining the Lagrangian function \( \mathcal{L}(F, \lambda_1, \lambda_2) = I(F) - \lambda_1 g_1(F) - \lambda_2 g_2(F) \), where \( g_1(F) = \int_0^A x \, dF(x) - \varepsilon \) and \( g_2(F) = E_{th} - \int_0^A b x \ln(1 + c x) \, dF(x) \). The necessary and sufficient conditions for optimality require that the derivative of the Lagrangian with respect to \( F \) is non-positive for all feasible distributions \( F \in \Omega \) and zero at the optimal distribution \( F^* \). Substituting the definitions of \( I(F) \), \( g_1(F) \), and \( g_2(F) \) into the Lagrangian, we obtain \( \frac{\partial \mathcal{L}(F, \lambda_1, \lambda_2)}{\partial F} = \int_0^A \left( i(x; F) - \lambda_1 x + \lambda_2 b x \ln(1 + c x) \right) dF(x) \). For the optimal distribution \( F^* \), this derivative must equal zero, while for all other distributions \( F \in \Omega \), it must satisfy \( \int_0^A ( i(x; F^*) - \lambda_1 x + \lambda_2 b x \ln(1 + c x) ) dF(x) \leq C - \lambda_1 \epsilon + \lambda_2 E_{th} \). Thus, the Lagrange multipliers \( \lambda_1 > 0 \) and \( \lambda_2 > 0 \) must be chosen such that this inequality holds for all \( F \in \Omega \), completing the proof. For further details, please refer to \cite{Smith_1971}.

\vspace{-4mm}
\section{Proof of Corollary 2}
 \label{ProofOfCorollary2}

The proof extends the necessary and sufficient conditions for optimality given by Theorem 2 to characterize the support of the optimal distribution \( F^* \).

First, recall that Theorem 2 establishes that \( F^* \) satisfies
\vspace{-1mm}
    \begin{align}
        \int_0^A ( i(x; F^*) &- \lambda_1 x + \lambda_2 b x \ln(1 + c x) ) dF(x) \nonumber \\
        &\leq C - \lambda_1 \varepsilon + \lambda_2 E_{th},
    \end{align}
    \vspace{-1mm}
where \( i(x; F^*) \) denotes the marginal information density.

Now, define the auxiliary functions $A_1(x) = x$, $A_2(x) = -b x \ln(1 + c x)$, $a_1 = \varepsilon$, and $a_2 = -E_{th}$. Using these definitions, Theorem 2 can be rewritten as            \begin{align}
        \int_0^A ( i(x; F^*) &- \lambda_1 A_1(x) + \lambda_2 A_2(x) ) dF(x) \nonumber \\
        &\leq C - \lambda_1 a_1 + \lambda_2 a_2.
    \end{align}

We need to show that this inequality holds, if and only if
    \begin{equation}
        i(x; F^*) \hspace{-0.5ex}\leq \hspace{-0.5ex}C + \lambda_1(A_1(x) - a_1) + \lambda_2(A_2(x) - a_2), \forall x \in [0, A],
    \end{equation}
    \vspace{-2mm}
and
\vspace{-2mm}
    \begin{equation}
        i(x; F^*) \hspace{-0.5ex}= \hspace{-0.5ex}C + \lambda_1(A_1(x) - a_1) + \lambda_2(A_2(x) - a_2), \hspace{-0.5ex} \forall x \hspace{-0.5ex}\in \hspace{-0.5ex}\text{Supp}(F^*).
    \end{equation}

To prove the "if" part, assume that the inequalities above hold. Then, by following similar steps as \cite{Smith_1971}. for any  distribution \( F \in \Omega \)
    \begin{align}
    &\int_0^A ( i(x; F^*) - \lambda_1 A_1(x) + \lambda_2 A_2(x) ) dF(x)\nonumber\\ 
    &\leq  C - \lambda_1 a_1 + \lambda_2 a_2.
    \end{align}

For the "only if" part, assume that the inequality in Theorem 2 holds but not the equality conditions. Then, there exists some \( x_0 \in \text{Supp}(F^*) \) such that
    \begin{equation}
        i(x_0; F^*) > C + \lambda_1(A_1(x_0) - a_1) + \lambda_2(A_2(x_0) - a_2).
    \end{equation}
Consider a step function distribution concentrated at \( x_0 \). The integral becomes
    \begin{align}
        \int_0^A \big( i(x_0; F^*) &- \lambda_1 A_1(x_0) + \lambda_2 A_2(x_0) \big) dF(x_0) \nonumber \\
        &> C - \lambda_1 a_1 + \lambda_2 a_2,
    \end{align}
which contradicts the original assumption. Therefore, \( i(x; F^*) \) must satisfy the equality conditions for all \( x \in \text{Supp}(F^*) \). This completes the proof.

\vspace{-3mm}
 \section{Proof of Theorem \ref{Theorem 3}}
 \label{ProofOfTheorem3}
This appendix derives the average channel transition probability $p_{Y|X}(y|x)$ for the lognormal fading channel, which is the cornerstone for calculating the mutual information in our ergodic, non-coherent capacity framework. This probability, $p(y|x)=\mathbb{E}_{h_t}[p(y|x,h_t)]$ with $C=\sup_{F \in \Omega} I(F)$.
 We show that the equality in \eqref{equ:fequal} can not be satisfied on a set of points that has an accumulation point, which indicates that the set $\text{Supp}(F^*)$ must be discrete and the optimal input $X$ must be a discrete random variable. We start with the necessary and sufficient conditions for the optimality of $F^*$. Extending the necessary and sufficient conditions of Corollary \ref{Corollary 2} to the complex domain, the LHS of \eqref{equ:fequal} is reduced to 
    \begin{align}\label{equ:s_z}
        s(z) &= \lambda_1(z-\varepsilon)-\lambda_2(bz\ln(1+cz)-E_{th})+C \nonumber \\
        &-\int p(y|z)\log_2\frac{{p(y|z)}}{{p(y;F^*)}}\d y, \quad z \in \mathcal{D}, 
    \end{align}
    where the domain $\mathcal{D}$ is defined by $\Re(z) > 0$. The function $s(z)$ is analytic over the complex domain, since the linear function, and the logarithmic function are all analytic. The necessary condition for the optimal input distribution $F^*$ is that $s(z)$ must be zero $\forall z \in \text{Supp}(F^*)$. However, from the identity theorem, if the set $\text{Supp}(F^*)$ has an accumulation point and the analytic function $s(z)=0, \forall z \in \mathcal{D}$, then $s(z)$ is necessarily zero over the whole $\mathcal{D}$, and hence, from \eqref{equ:s_z}, we have
        \begin{align} \label{equ:logy}
        &\int p(y|z)\log_2{p(y;F^*)}\d y = -C -\lambda_1(z-\varepsilon)\nonumber \\
        &+\lambda_2(bz\ln(1+cz)-E_{th})+\int p(y|z)\log_2{p(y|z)}\d y.
        \end{align}

\subsection{Conditional Probability of the Channel}
The conditional probability $p_{Y|X}(y|x)$ of the lognormal channel is not present in the literature. Hence, we derive $p_{Y|X}(y|x)$, which will be used to derive the properties of the optimal input distribution. According to~\eqref{equ:e1}, the random variable $Y$ of the channel output can be represented as $Y = X_1 + N$, where $X_1 = aR_Ph_{1,l}H_tX$, $H_t$ is the lognormal fading, $X$ is the channel input, and $N$ is the AWGN. By using the basic properties of pdfs, the pdf of $Y$ is the convolution of pdfs of $X_1$ and $N$. Hence, for a given input $X$
    \begin{align}\label{equ:convolution}
        p_{Y|X}(y|x) = \int_{0}^{+\infty} p_{X_1|X}(t|x) p_{N}(y-t)\d t,
    \end{align}
where $p_{X_1|X}(t|x)$ is the pdf of $X_1$ for a given input $X$, and $p_N(t)$ is the pdf of $N$, respectively. By using the variable transformation from $H_t$ to $X_1$ in~\eqref{equ:logn}, $p_{X_1|X}(x_1|x)$ can be deduced, and hence, the conditional probability in~\eqref{equ:convolution} is written as
    \begin{align}
        \label{equ:convolution2}
        &p_{Y|X}(y|x) = \frac{1}{4\pi\sigma_g\sigma_{X_l}}\int_{0}^{+\infty} \frac{1}{t} \\
        &\times\exp{-\frac{\left(\ln {\left(\frac{t}{aR_Ph_{1,l}x}\right)}-2\mu_{X_l}\right)^2}{8\sigma_{X_l}^2}}\exp{-\frac{(y-t)^2}{2\sigma_g^2}}\d t. \nonumber
    \end{align}
By apply the variable transformation $k = \ln\left(\frac{t}{aR_Ph_{1,l}x}\right)$ and after several manipulations, we obtain
    \begin{align}\label{equ:convolution3}
        &p_{Y|X}(y|x) = \frac{e^{-\frac{y^2}{2\sigma_g^2}}}{4\pi\sigma_g\sigma_{X_l}}\int_{-\infty}^{+\infty} \exp\left(-\frac{\left(k-2\mu_{X_l}\right)^2}{8\sigma_{X_l}^2}\right) \nonumber \\  &\times\exp\left(\frac{aR_Ph_{1,l}xe^k}{\sigma_g}\frac{y}{\sigma_g}-\left(\frac{aR_Ph_{1,l}xe^k}{\sigma_g}\right)^2\right)\d k.
    \end{align}
It can be identified that the second exponential term inside the integration is in the form of the generating function of the Hermite polynomial which is $e^{xt-\frac{t^2}{2}}=\sum_{n=0}^{\infty}H_{e_n}(x)\frac{t^n}{n!}$, where $H_{e_n}(x)$ is the Hermite polynomial of $n$-th kind~\cite{Fah_2012}. Applying this relation, interchanging the order of summation with integration, and with some rearrangements, \eqref{equ:convolution3} can be written as
    \begin{align}\label{equ:convolution4}
        \hspace{-0.5ex}p_{Y|X}(y|x) &= \frac{e^{-\frac{y^2}{2\sigma_g^2}}}{4\pi\sigma_g\sigma_{X_l}}\sum_{n=0}^{\infty}\frac{1}{n!}H_{e_n}\left(\frac{y}{\sigma_g}\right)\left(\frac{aR_Ph_{1,l}x}{\sigma_g}\right)^n \nonumber \\
        &\times\int_{-\infty}^{+\infty} \exp{\left(-\frac{\left(k-2\mu_{X_l}\right)^2}{8\sigma_{X_l}^2}+nk\right)}\d k.
    \end{align}
The integration in \eqref{equ:convolution4} can be identified as a standard result of the integration as follows
    \begin{equation}
    \int_{-\infty}^{\infty}\exp{(-ak^2+bk-c)\d k} = \sqrt{\frac{\pi}{a}}\exp{\left(\frac{b^2}{4a}-c\right)},
    \end{equation}
where $a = \frac{1}{8\sigma_{X_l}^2}$, $b = \frac{\mu_{X_l}+2n\sigma_{X_l}^2}{2\sigma_{X_l}^2}$, and $c = \frac{\mu_{X_l}^2}{2\sigma_{X_l}^2}$. With simplifications, the final expression can be expressed as
 \begin{align}\label{equ:conditional_f}
        p_{Y|X}(y|x) = \frac{e^{-\frac{y^2}{2\sigma_g^2}}}{\sqrt{2\pi\sigma_g^2}}\sum_{n=0}^{\infty}\frac{1}{n!}  H_{\mathrm{e}_n}\left(\frac{y}{\sigma_g}\right) K_{n}^n x^n,
    \end{align}
where $K_n=aR_Ph_{1,l}e^{2(n\sigma_{X_l}^2+\mu_{X_l})}$. Similar to \cite{Fah_2012}, we set $\sigma_g^2=1$ to simplify the proof without the loss of generality. 

Now, $\log_2\left(p(y;F^*)\right)$ in \eqref{equ:logy} is a continuous function of $y$ and is a square integrable with respect to $e^{-y^2/2}$. As such, it can be written in terms of the Hermite polynomials as 
\vspace{-3mm}
\begin{equation}
\label{eq:pdf}
\log_2\left(p(y;F^*)\right)=\sum_{m=0}^{\infty}c_m   H_{\mathrm{e}_m}(y).
\end{equation}
Next, by using \eqref{equ:conditional_f}, \eqref{eq:pdf}, and the orthogonality property of the Hermite polynomials with respect
to $e^{-y^2/2}$ i.e., $\int_{-\infty}^{\infty} e^{-y^2/2} H_{\mathrm{e}_m}(y) H_{\mathrm{e}_n}(y)=m!\sqrt{2 \pi}$,
if $m=n$ zero and otherwise~\cite{Morsi}, $\int p(y|z)\log{p(y;F^*)}\d y$ is expressed as
\vspace{-3mm}
\begin{equation}\label{equ:p_yF}
    \int p(y|z)\log_2{p(y;F^*)}\d y=\sum_{m=0}^{\infty}c_m (K_mz)^m. 
\end{equation}
Let $V_n=f^{(n)}[bx\ln(1+cx)](0)$. Then, the Taylor expansion of $bz \ln(1+cz)$ is given by
\begin{equation}\label{equ:plog}
    bz \ln(1+cz)=\sum_{m=0}^{\infty}V_mz^m.
\end{equation}
Similarly, $\log_2\left(p(y|z)\right)$ is a continuous function of $y$ and is square integrable with respect to $e^{-y^2/2}$, and hence, it can be written in terms of the Hermite polynomials as 
\begin{equation}\label{equ:log_py}
    \log_2\left(p(y|z)\right)=\sum_{m=0}^{\infty}a_m   H_{\mathrm{e}_m}(y).
\end{equation}
Hence, with the use of \eqref{equ:conditional_f}, \eqref{equ:log_py}, and the orthogonal property of Hermite polynomial, the term $\int p(y|z)\log_2{p(y|z)}\d y$ is written as
\begin{equation}\label{equ:pyz}
    \int p(y|z)\log_2{p(y|z)}\d y=\sum_{m=0}^{\infty}a_m (K_mz)^m. 
\end{equation}
With the use of \eqref{equ:p_yF}, \eqref{equ:plog}, and \eqref{equ:pyz}, the expression in \eqref{equ:s_z} is reduced to 
\begin{align}
    \sum_{m=0}^{\infty}c_m (K_mz)^m&=\lambda_2\left(\sum_{m=0}^{\infty}V_mz^m-E_{th}\right)-\lambda_1\left(z-\varepsilon\right) \nonumber \\
    &-C+\sum_{m=0}^{\infty}a_m (K_mz)^m.
\end{align}
Equating the coefficients of $z^m$, we have

\vspace{-3mm}
\begin{IEEEeqnarray}{ccl}
    c_0&=&\lambda_2\left(V_0-E_{th}\right)+\lambda_1\varepsilon-C+a_0,\\
    c_1&=&\frac{\lambda_2V_1-\lambda_1+a_1K_1}{K_1},\\
    \label{eq:coef}
    c_m&=&\frac{\lambda_2V_m+a_mK_m^m}{K_m^m}, \hspace{2ex} \text{for all $m >1$}.
\end{IEEEeqnarray}
By inserting \eqref{eq:coef} into \eqref{eq:pdf}, the following holds
\begin{equation}
    p(y;F^*)=e^{\ln(2)\sum_{m=0}^{\infty}c_m   H_{\mathrm{e}_m}(y)}.
\end{equation}
Next, by following similar steps as \cite{Morsi}, it can be shown that $p(y;F^*)$ can not be a valid output distribution. Hence, $\text{Supp}(F^*)$ can not have an accumulation point and $X^*$ must be a discrete random variable. This completes the proof.

\end{appendices}

\vspace{-3mm}
\bibliographystyle{IEEEtran}
\bibliography{IEEEabrv,Main}

\begin{IEEEbiography}[{\includegraphics[width=1in,keepaspectratio]{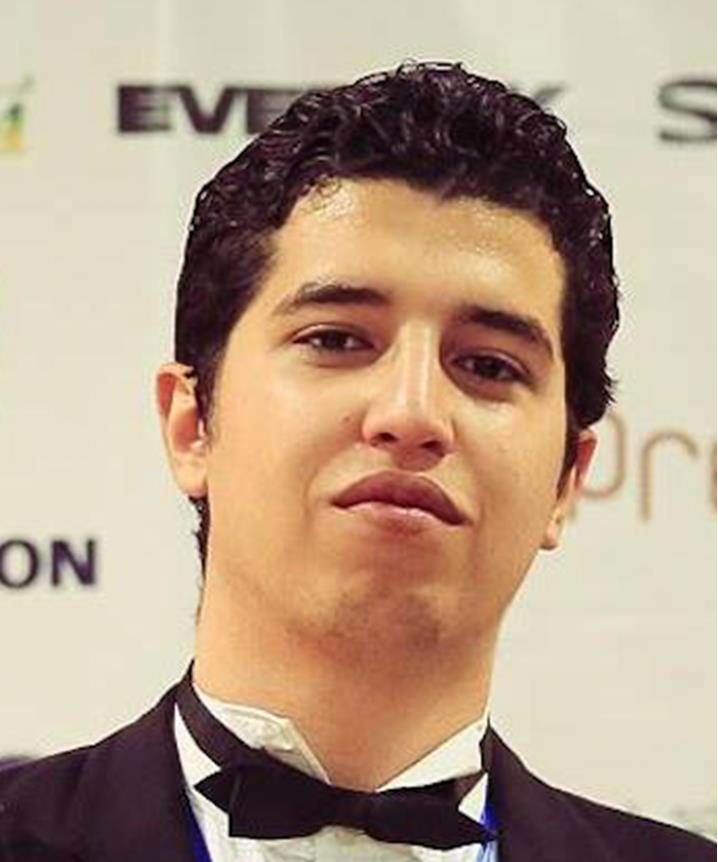}}]{Nizar Khalfet} (Member, IEEE) received the B.Sc.
degree in electrical engineering from Sup’Com,
Tunisia, the M.Sc. degree from CentraleSupélec,
Paris, France, in 2015, and the Ph.D. degree in
electrical engineering from the Institut National
de Recherche en Informatique et en Automatique
(INRIA), France. Since 2020, he has been a
Researcher with the IRIDA Research Centre for
Communication Technologies, University of Cyprus.
His research interests include the intersection of
information theory, intelligent reflected surfaces, and
communications theory.
\end{IEEEbiography}

\begin{IEEEbiography}[{\includegraphics[width=1in,keepaspectratio]{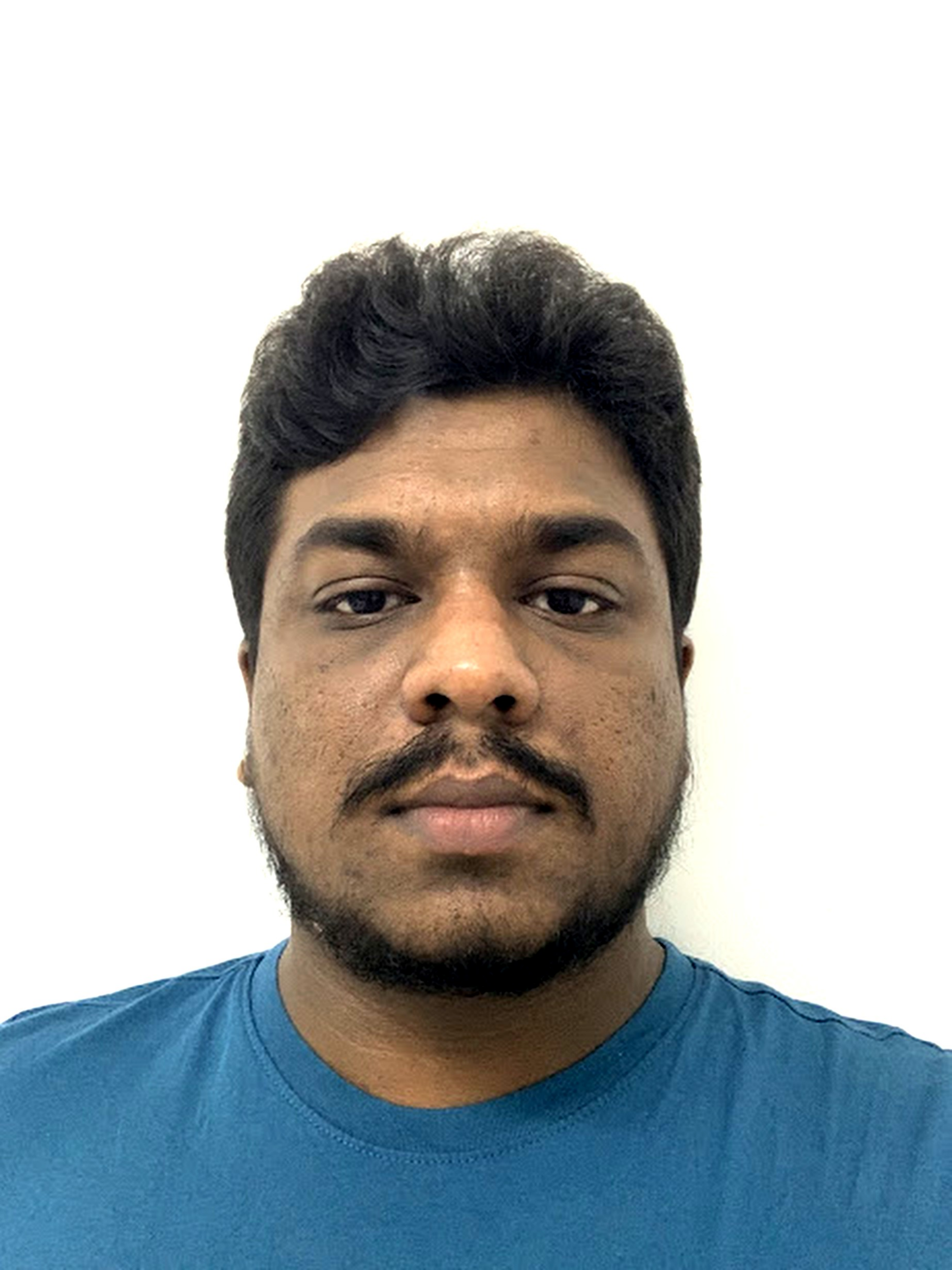}}]{Kapila W. S. Palitharathna}
(Member, IEEE) received the B.Sc. Eng. (Hons.) and Ph.D. degrees in Electrical and Electronic Engineering from the University of Peradeniya, Sri Lanka, in 2016, and in 2022. He is currently a Postdoctoral Research Associate at the IRIDA Research Centre for Communication Technologies, Department of Electrical and Computer Engineering, University of Cyprus. He has worked as a Research Assistant, a Senior Lecturer, and the Head of the Department (Telecommunication) at the Sri Lanka Technological Campus, Sri Lanka, from 2018 to 2023. He also worked as an Instructor at the University of Peradeniya, Sri Lanka, from 2016 to 2017. His research interests include visible light communication, underwater optical wireless communication, non-orthogonal multiple access, intelligent reflective surfaces, energy harvesting communications, and machine learning for communication. 
\end{IEEEbiography}

\begin{IEEEbiography}[{\includegraphics[width=1in,keepaspectratio]{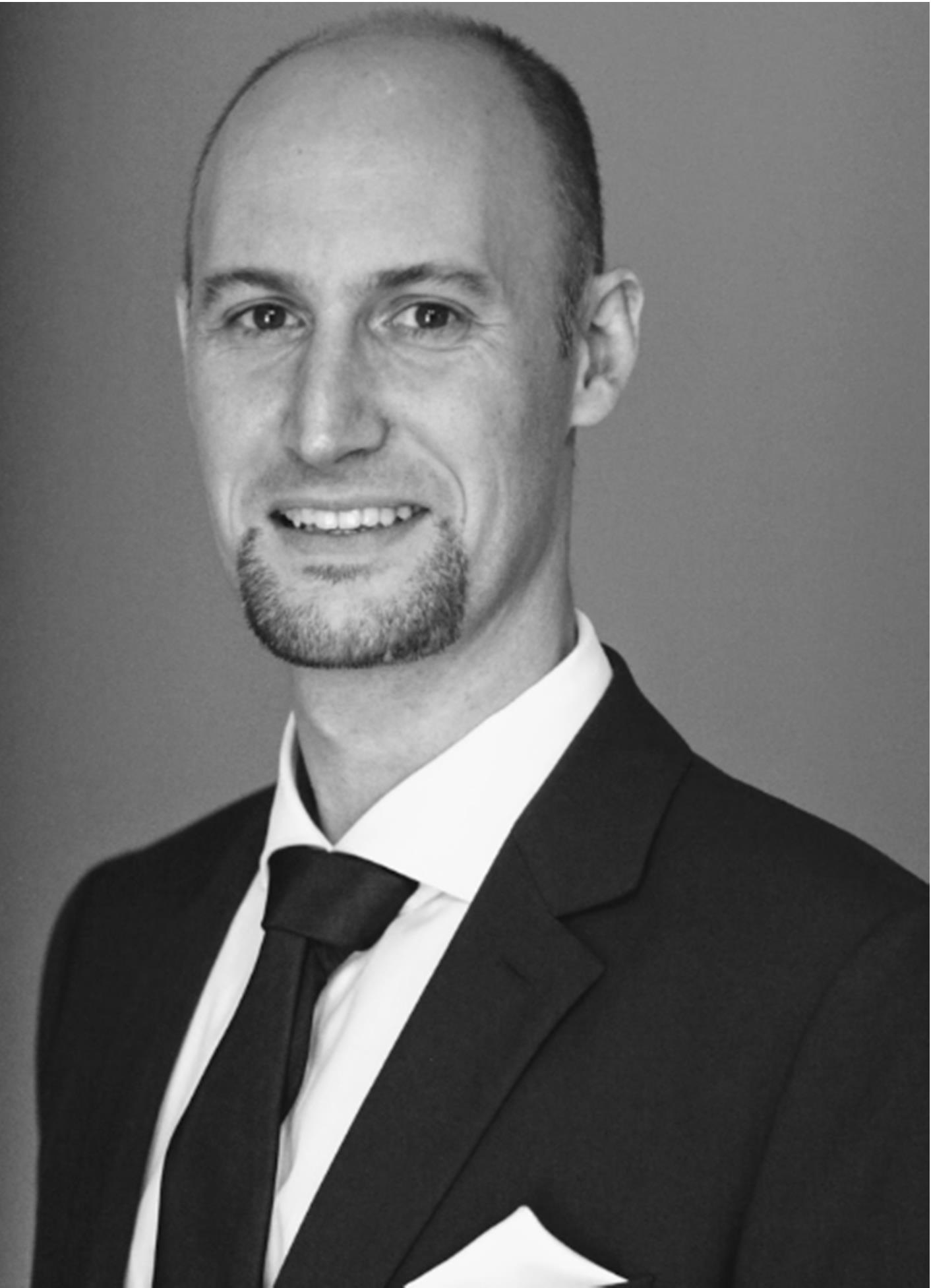}}]{Symeon Chatzinotas} (Fellow, IEEE) received the M.Eng. in Telecommunications from Aristotle University of Thessaloniki, Greece and the M.Sc. and
Ph.D. in Electronic Engineering from University of Surrey, UK in 2003, 2006 and 2009 respectively. He is currently Full Professor/Chief Scientist I and Head of the research group SIGCOM in the Interdisciplinary Centre for Security, Reliability and Trust, University of Luxembourg. In parallel, he is an Adjunct Professor in the Department of Electronic Systems, Norwegian University of Science and Technology, an Eminent Scholar of the Kyung Hee University, Korea and a Collaborating Scholar of the Institute of Informatics \& Telecommunications, National Center for Scientific Research ”Demokritos”. In the past, he has been a Visiting Professor at EPFL, Switzerland and University of Parma, Italy and contributed in numerous R \& D projects for the Institute of Telematics and Informatics, Center of Research and Technology Hellas and Mobile Communications Research Group, Center of Communication Systems Research, University of Surrey. He has authored more than 800 technical papers in refereed international journals, conferences and scientific books and has received numerous awards and recognitions, including the IEEE Fellowship, IEEE Distinguished Contributions Award and IEEE Harry Rowe Mimno Award. He has served in the editorial board of npj Wireless Technology, IEEE Transactions on Communications, IEEE Open Journal of Vehicular Technology and the International Journal of Satellite Communications and Networking.
\end{IEEEbiography}

\begin{IEEEbiography}[{\includegraphics[width=1in,keepaspectratio]{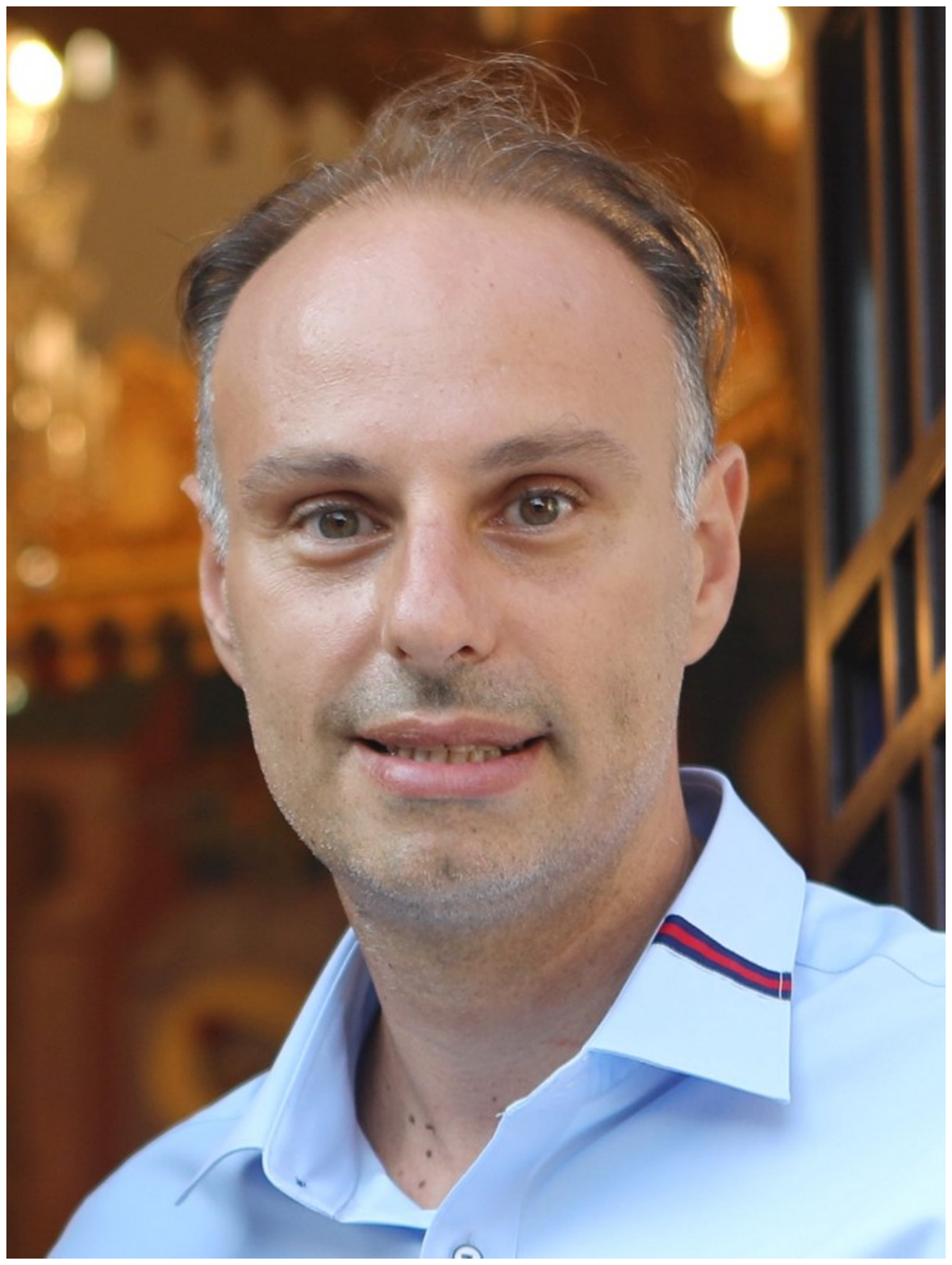}}]{Ioannis Krikidis} (F’19) received the diploma in Computer Engineering from the Computer Engineering and Informatics Department (CEID) of the University of Patras, Greece, in 2000, and the M.Sc and Ph.D degrees from École Nationale Supérieure des Télécommunications (ENST), Paris, France, in 2001 and 2005, respectively, all in Electrical Engineering. From 2006 to 2007 he worked, as a Post-Doctoral researcher, with ENST, Paris, France, and from 2007 to 2010 he was a Research Fellow in the School of Engineering and Electronics at the University of Edinburgh, Edinburgh, UK. He is currently a Professor at the Department of Electrical and Computer Engineering, University of Cyprus, Nicosia, Cyprus. His current research interests include wireless communications, quantum computing, 6G communication systems, wireless powered communications, and intelligent reflecting surfaces. He serves as an Area Editor for IEEE Transactions on Communications and as Editor-in-Chief of Frontiers in Communications and Networks. He was the recipient of the Young Researcher Award from the Research Promotion Foundation, Cyprus, in 2013, and the recipient of the IEEE ComSoc Best Young Professional Award in Academia, 2016, and IEEE Signal Processing Letters best paper award 2019. He has been recognized by the Web of Science as a Highly Cited Researcher for 2017-2021. He has received the prestigious ERC Consolidator Grant for his work on wireless powered communications.
\end{IEEEbiography}

\end{document}